\documentclass[letterpaper, 10pt, conference]{ieeeconf}
\IEEEoverridecommandlockouts
\usepackage{textcomp}
\usepackage[utf8]{inputenc}
\usepackage[T1]{fontenc}
\usepackage{url}
\usepackage{booktabs}
\usepackage{amsmath,amssymb,amsfonts}

\usepackage{amsthm}
\usepackage{mathtools}
\usepackage{algorithmic}
\usepackage{algorithm}
\usepackage{graphicx}
\usepackage{float}
\usepackage{placeins}
\usepackage{dblfloatfix}
\usepackage{subcaption}
\usepackage{caption}
\usepackage{xcolor}
\usepackage{microtype}
\let\labelindent\relax
\usepackage{enumitem}
\usepackage{siunitx}
\usepackage{nicefrac}

\graphicspath{{figures/}}

\AtBeginDocument{%
  \abovedisplayskip=4pt plus 2pt minus 2pt
  \belowdisplayskip=4pt plus 2pt minus 2pt
  \abovedisplayshortskip=2pt plus 1pt
  \belowdisplayshortskip=2pt plus 1pt
}
\setlist{nosep,leftmargin=*}
\makeatletter
\let\@afterindenttrue\@afterindentfalse
\makeatother
\DeclareRobustCommand{\Seta}{\hyperref[eq:stopped_cost_exact_compact]{\mathcal{L}_\eta}}
\DeclareRobustCommand{\Setah}{\hyperref[eq:stopped_cost_num_compact]{\mathcal{L}_{\eta,h}}}
\DeclareRobustCommand{\ePhis}{\hyperref[eq:strong_indicator_error]{e_{\Phi}^{\mathrm{s}}}}
\DeclareRobustCommand{\ePhiw}{\hyperref[eq:weak_indicator_error]{e_{\Phi}^{\mathrm{w}}}}
\DeclareRobustCommand{\krate}{\hyperref[eq:leading_rate]{\kappa}}
\DeclareRobustCommand{\arate}{\hyperref[eq:Lq_exit_rate]{\alpha}}
\DeclareRobustCommand{\brate}{\hyperref[eq:terminal_tube]{\beta}}
\DeclareRobustCommand{\grate}{\hyperref[eq:mesh_point_rate]{\gamma}}

\newtheorem{sassumption}{\textbf{}}

\newtheorem{theorem}{Theorem}
\newtheorem{remark}{\textbf{Remark}}

\newtheorem{corollary}{Corollary}

\usepackage{eso-pic}
\usepackage[hidelinks]{hyperref}
\makeatletter
\def\@citex[#1]#2{\leavevmode
  \let\@citea\@empty
  \@cite{\@for\@citeb:=#2\do
    {\@citea\def\@citea{],\penalty\@m\ [}%
     \edef\@citeb{\expandafter\@firstofone\@citeb\@empty}%
     \if@filesw\immediate\write\@auxout{\string\citation{\@citeb}}\fi
     \hyperlink{cite.\@citeb}{\@ifundefined
       {b@\@citeb}{\hbox{\reset@font\bfseries ?}%
       \G@refundefinedtrue
       \@latex@warning
         {Citation `\@citeb' on page \thepage \space undefined}}%
       {\csname b@\@citeb\endcsname}}}}{#1}}
\let\oldbibitem\bibitem
\renewcommand{\bibitem}[1]{\oldbibitem{#1}\hypertarget{cite.#1}{}}
\makeatother

\DeclareFontFamily{OMS}{cmsy}{\skewchar\font48 }
\DeclareFontShape{OMS}{cmsy}{m}{n}{%
      <5>gen*cmsy%
      <6>cmsy7%
      <7><8><9><10>gen*cmsy%
      <10.95><12><14.4><17.28><20.74><24.88>cmsy10%
      }{}
\DeclareFontShape{OMS}{cmsy}{b}{n}{%
      <5><6><7><8><9>gen*cmbsy%
      <10><10.95><12><14.4><17.28><20.74><24.88>cmbsy10%
      }{}
\begin{document}
\title{\LARGE \bf Higher-Order Approximation of Exit Functionals in Sampling-Based Stochastic Model Predictive Control}

\newcommand{\fundingstatement}{This work is supported by DARPA COMPASS HR0011-25-3-0210 and AFOSR DSCT FA9550-25-1-0347.}

\author{Sashank Modali and Takashi Tanaka%
\thanks{\fundingstatement\
        S.~Modali and T.~Tanaka are with the School of Aeronautics and Astronautics,
        Purdue University, IN 47906.
        {\tt\small \{nmodali, tanaka16\}@purdue.edu}}%
}

\maketitle
\AddToShipoutPictureBG*{\AtPageLowerLeft{%
  \hspace*{\dimexpr 1in+\oddsidemargin\relax}%
  \raisebox{22pt}[0pt][0pt]{\parbox[b]{\textwidth}{\scriptsize
  \textcopyright{} 2026 IEEE. Personal use of this material is permitted. Permission from IEEE must be obtained for all other uses, in any current or future media, including reprinting/republishing this material for advertising or promotional purposes, creating new collective works, for resale or redistribution to servers or lists, or reuse of any copyrighted component of this work in other works.}}}}
\thispagestyle{empty}
\pagestyle{empty}
\begin{abstract}
Safety evaluation in sampling-based stochastic model predictive control often requires numerical estimation of exit functionals. The approximation of first-exit times and exit indicators is therefore a key numerical bottleneck, and discretization error in these quantities directly affects the resulting controller. This paper studies how existing higher-order methods for strong approximation of exit times can be brought into safe control. Two cases are highlighted. For general noncommutative dynamics, an adaptive order-1 Milstein discretization is used together with L\'evy-area simulation via Wiktorsson's method. For commutative dynamics, an adaptive order-1.5 construction achieves a stronger exit-time rate. Under a local anti-concentration condition on the exit-time law, we show that strong exit-time approximation transfers to strong approximation of the failure indicator. The methods are then studied in the context of chance-constrained path integral control, which provides an exact continuous-time representation of safety through exit events. Numerical experiments compare the two cases in terms of strong exit-time error, failure-indicator error, and closed-loop constraint satisfaction, showing improvement over Euler--Maruyama and thereby enabling existing and future techniques whose applicability depends on improved strong approximation.\end{abstract}

\section{INTRODUCTION}
\label{sec:introduction}

A system that underestimates its risk of failure cannot be trusted to operate safely. Many sampling-based stochastic control methods~\cite{schildbach2014scenario,lew2024risk} enforce safety by sampling trajectories and evaluating exit-events over a prescribed safe set. In path integral control~\cite{kappen2005path,theodorou2010generalized}, chance-constrained formulations~\cite{patil2022chance,patil2025strong} can also represent the optimal failure probability exactly using exit-event indicators. The safety guarantee is therefore only as accurate as the numerical approximation of this indicator. For continuous-time dynamics, evaluating it requires time discretization (step size $h$) of the exit event, which is discontinuous and sensitive to boundary crossing~\cite{giles2018stopped,giles2023discontinuous}.

Under Euler--Maruyama, a discretized trajectory may cross the boundary between mesh points without triggering exit, resulting in decreased accuracy. Gobet and Menozzi~\cite{gobet2010stopped} recover first-order weak accuracy ($O(h)$ exit functional bias) via an $O(\sqrt{h})$-boundary correction, but the strong accuracy (pathwise approximation error of exit functional) remains at $O(\sqrt{h})$. Strong convergence governs the variance of multi-level estimators~\cite{giles2018stopped}. It is rate-limiting, and thus indispensable for accurate risk estimation in safety-critical applications.

Recently, Hoel and Raghunathan (HR)~\cite{hoel2024adaptive} proposed two strong It\^o--Taylor methods that adaptively refine $h$ near the boundary, achieving strong exit-time rates $O(h^{1-\xi})$ and $O(h^{3/2-\xi})$, respectively, for all $\xi>0$, under a restrictive commutativity condition on the noise. However, two issues remain. First, the object entering
the safety specification in control is not the exit time itself but
an exit indicator or a stopped cost functional, and the
transfer of strong error rates to these quantities has not been addressed. Second, stochastic control routinely
involves noncommutative diffusions, for which L\'evy-area
simulation is required but exit-time convergence has not
been established. This paper bridges both gaps and
establishes consequences for closed-loop safe control.

The main contributions are as follows.
\begin{enumerate}
    \item We identify a terminal boundary-layer condition under which exit-time and state approximation rates transfer to the failure indicator and the relaxed stopped functional. The transfer is method-agnostic, depending only on the exit-time rate, state rate, and a terminal tube exponent $\brate$.

    \item We bridge the adaptive higher-order methods of~\cite{hoel2024adaptive} into safe stochastic optimal control. For noncommutative dynamics, the order-\(1\) Milstein discretization with Wiktorsson L\'evy areas~\cite{wiktorsson2001joint} achieves exit-time rate \(O(h^{1-\xi})\) and indicator rate \(O(h^{(1-\xi){\brate}/(1+{\brate})})\). For commutative dynamics, the order-\(1.5\) construction achieves \(O(h^{3/2-\xi})\) and \(O(h^{(3/2-\xi){\brate}/(1+{\brate})})\), yielding \(3/4-\varepsilon\) at \(\brate=1\). These rates carry over to the relaxed stopped cost.

    \item We specialize to chance-constrained path integral control (CC-PIC) and show that the closed-loop failure probability of the optimal policy inherits the same improved rates. Experiments on a controlled unicycle confirm the predicted orders and demonstrate tighter constraint satisfaction compared to Euler--Maruyama (EM).
\end{enumerate}

The paper is organized as follows. Section~\ref{sec:problem} formulates the problem. Section~\ref{sec:discretization_error} establishes the transfer from exit-time error to indicator and stopped-functional error. Sections~\ref{sec:higher_order_methods}--\ref{sec:ccpic} present the higher-order methods and their specialization to path-integral control. Section~\ref{sec:results} reports experiments.

\subsection{Notation.} 
Capital italic letters denote random variables ($\boldsymbol X_t$, $\boldsymbol W_t$), bold letters denote matrices and vectors ($\mathbf F$, $\mathbf B$, $\mathbf G$), and superscript $\eta$ denotes dependence on the multiplier ($V^\eta$, $\psi^\eta$). Numerical counterparts carry a discretization superscript $h$ ($\bar{\boldsymbol X}^h$, $\bar Y^h$). The horizon $T$ is deterministic. $\overline D$ denotes the closure of $D$. 
\section{PROBLEM FORMULATION}
\label{sec:problem}

We consider a controlled diffusion on $(\Omega,\mathcal F,\{\mathcal F_t\}_{t\ge 0},\mathbb P)$,
driven by an \(m\)-dimensional Wiener process \(\boldsymbol W_t\). The state
\(\boldsymbol X_t\in\mathbb R^d\) evolves according to
\begin{align}
\label{eq:controlled_sde}
d\boldsymbol X_t
=
\mathbf F(\boldsymbol X_t,\mathbf u_t)\,dt
+
\mathbf B(\boldsymbol X_t)\,d\boldsymbol W_t,
\end{align}
with initial condition \(\boldsymbol X_0=\mathbf x_0\in D\), where \(\mathbf u_t\in\mathbb R^p\) is an admissible control process,
\(\mathbf F:\mathbb R^d\times\mathbb R^p\to\mathbb R^d\) is the controlled drift,
\(\mathbf B:\mathbb R^d\to\mathbb R^{d\times m}\) is the diffusion coefficient,
\(D\subset\mathbb R^d\) is the prescribed safe set (see Section~\ref{sec:regularity} for regularity conditions), and \(T>0\) is the time horizon.

For \(r>0\), define the tubular neighborhood \(D_r:=\{\mathbf x\in\mathbb R^d:\operatorname{dist}(\mathbf x,D)<r\}\). Following~\cite{hoel2024adaptive}, fix \(\bar R_D>0\), a scalar, and let \(\boldsymbol\mu:D\to\mathbb R^p\) be a time-invariant nominal feedback, so that the closed-loop drift \(\mathbf F^{\mu}(\boldsymbol x):=\mathbf F(\boldsymbol x,\boldsymbol\mu(\boldsymbol x))\) is time-homogeneous. The optimal feedback is denoted by \(\boldsymbol\mu^*_\eta:\overline D \times[0,T]\to\mathbb R^p\) for each multiplier \(\eta\ge 0\) introduced in \eqref{eq:lagrangian_compact}. The following conditions are imposed on the nominal closed-loop coefficients.

\subsection{Standing regularity conditions}
\label{sec:regularity}

The following conditions are those of~\cite[Assumptions~A and~B]{hoel2024adaptive}, under which the exit-time strong approximation rates $\mathcal{O}(h^{1-\xi})$ and $\mathcal{O}(h^{3/2-\xi})$ are established.

\begin{enumerate}
\item
\(D\subset\mathbb R^d\) is bounded and \(\partial D\) is \(C^4\), with outward unit normal \(n_D\in C^3\), uniform ball property, and $\pi_r(\mathbf y)=\mathbf y+r\,n_D(\mathbf y)$ a \(C^3\)-diffeomorphism $\partial D\to \partial D_r$ for every \(r\in[0,\bar R_D]\).

\item
The nominal closed-loop drift \(\mathbf F^{\mu}(\boldsymbol x)=\mathbf F(\boldsymbol x,\boldsymbol\mu(\boldsymbol x))\) and \(\mathbf B(\boldsymbol x)\) are \(C^3\) on \(D_{\bar R_D}\) and admit \(C^3\) extensions to an open neighborhood of \(\overline{D_{\bar R_D}}\).

\item
(Uniform ellipticity.) There exist \(0<\hat c_b\le\bar C_b\) such that $\hat c_b |\xi|^2 \le \xi^\top (\mathbf B\mathbf B^\top)\xi \le \bar C_b |\xi|^2$ for all \((\mathbf x,\xi)\in D_{\bar R_D}\times\mathbb R^d\).
\end{enumerate}

\subsection{Exit time and safety functional}

Safety is represented through the first exit time from the safe set,
\begin{equation}
\label{eq:raw_exit_time}
\tau_D := \inf\{t\ge 0 : \boldsymbol{X}_t \notin D\},
\end{equation}
with \(\tau_D=\infty\) if the process never exits \(D\). The failure indicator and probability under control \(\mathbf u\) are
\begin{align}
\label{eq:safety_functional}
\Phi_T(\boldsymbol{X}) &:= \mathbf 1\{\tau_D<T\}, \\
\label{eq:failure_probability}
p(\mathbf u) &:= \mathbb E\big[\Phi_T(\boldsymbol X^{\mathbf u})\big] = \mathbb P(\tau_D^{\mathbf u}<T)
\end{align}
Define the clipped exit time $\tau := \tau_D\wedge T$ and the stopped state $\boldsymbol X_\tau := \boldsymbol X_{\tau_D\wedge T}$. Following~\cite{hoel2024adaptive}, strong approximation results are stated for $\tau$, while $\tau_D$ is required to define $\Phi_T$.

\subsection{Safe stochastic optimal control}
Safety can be enforced through various constraint formulations, including expectation constraints, event-probability bounds, and risk measures. In this paper, we restrict attention to the subclass where the fixed Lagrange multiplier problem admits a classical HJB characterization~\cite{PfeifferTanZhou2021,BouchardNutz2012}. Let \(\ell(\mathbf{x},\mathbf{u},t)\) be the running cost and \(\phi_1,\phi_2\) be terminal penalties for failure and survival, respectively. The stopped cost is
\begin{align}
\label{eq:stopped_cost}
J(\mathbf{u})
=
\mathbb E\Big[
&\phi_1\,\mathbf 1\{\tau_D<T\}
+
\phi_2(\boldsymbol{X}_T)\,\mathbf 1\{\tau_D\ge T\} \nonumber\\
&+ \int_0^{\tau} \ell(\boldsymbol{X}_s,\mathbf{u}_s,s)\,ds
\Big]
\end{align}
The safety-constrained stochastic optimal control problem is
\begin{equation}
\label{eq:generic_safe_soc}
\begin{aligned}
\min_{\mathbf{u}}\quad & J(\mathbf{u}) \\
\text{s.t.}\quad & \mathcal S(\mathbf{u}) \le 0
\end{aligned}
\end{equation}
where \(\mathcal S(\mathbf{u}) :=\mathbb E[ \Gamma(\boldsymbol X^{\mathbf u})]-\Delta\) is a trajectory-level safety functional with \( \Gamma\) having the same structure as~\eqref{eq:stopped_cost} and a prescribed threshold \(\Delta>0\).
Introducing a multiplier \(\eta\ge 0\) for the constraint \(\mathcal S(\mathbf u)\le 0\), the Lagrangian is
\begin{equation}
\label{eq:lagrangian_compact}
\mathcal L(\mathbf u,\eta)
:=
J(\mathbf u)+\eta\,\mathcal S(\mathbf u)
\end{equation}
The dual function is $g(\eta):=\inf_{\mathbf u}\mathcal L(\mathbf u,\eta)$ and the dual problem is $\sup_{\eta\ge 0}g(\eta)$. For each fixed $\eta$, the inner minimization is an unconstrained stopped-cost problem. The constraint penalty $\eta\,\mathcal S(\mathbf u)$ is absorbed into augmented terminal and running costs $\psi^\eta$ and $\ell^\eta$. The augmented terminal cost is
\begin{equation}
\label{eq:augmented_terminal}
\psi^\eta = \psi_1^\eta\,\mathbf 1\{\tau_D<T\} + \psi_2^\eta(\boldsymbol X_T)\,\mathbf 1\{\tau_D\ge T\}
\end{equation}
where $\psi_i^\eta=\phi_i+\eta\,\Gamma_i$ for $i=1,2$, with $\Gamma_1, \Gamma_2$ the exit and survival components of $\Gamma$. The value function, denoted by $V^\eta:=V(\mathbf x,t;\eta)$, is given by:
\begin{equation}
\label{eq:value_function}
V^\eta = \inf_{\mathbf{u}}\,
\mathbb E\Big[
\psi^\eta
+\!\int_t^{\tau}\!\ell^\eta(\boldsymbol{X}_s,\mathbf{u}_s,s)\,ds
\;\Big|\;
\boldsymbol{X}_t\!=\!\mathbf{x}
\Big]
\end{equation} The value function is assumed to be a classical solution of the exit-time Hamilton--Jacobi--Bellman equation
\begin{equation}
\label{eq:hjb}
-\partial_t V^\eta
=
\inf_{\mathbf{u}}\big\{
\ell^\eta + \mathbf{F}^\top \partial_{\mathbf{x}} V^\eta
+
\tfrac{1}{2}\operatorname{Tr}\!\big(
\mathbf{B}\mathbf{B}^\top
\partial_{\mathbf{x}}^2 V^\eta
\big)
\big\}
\end{equation}
on \(D\times[0,T)\), with boundary condition \(V^\eta=\psi^\eta\) on \(\partial D\times[0,T)\cup \overline D\times\{T\}\). The associated minimizing feedback \(\boldsymbol\mu^*_\eta(\mathbf x,t):=\mathbf u^*(\mathbf x,t;\eta)\) is taken to be admissible, with the induced closed-loop drift Lipschitz in \(\boldsymbol x\) uniformly in \(t\). This ensures well-posedness of the closed loop diffusion needed for strong duality and classical HJB synthesis (Assumption~\ref{ass:classical}).

Specializing to the survival event \(\mathcal E^{\mathbf u}=\{\tau_D\ge T\}\) gives \(\mathcal S(\mathbf u)=p(\mathbf u)-\Delta\) with \(\Delta\in(0,1)\), so the Lagrangian becomes \(\mathcal L(\mathbf u,\eta)=J(\mathbf u)+\eta(p(\mathbf u)-\Delta)\). We denote $\psi^\eta$ as $\varphi^\eta$ in this case, with $\Gamma_1=1$ and $\Gamma_2=0$.
\subsection{Numerical approximation}
\label{sec:numerical_approx}

In numerical implementation, the value function~\eqref{eq:value_function}, the dual function $g(\eta)$, and the failure probability~\eqref{eq:failure_probability} are evaluated from discretized sampled trajectories of~\eqref{eq:controlled_sde}. Let\label{eq:num_state_def} $\bar{\boldsymbol X}^{h}$ denote a time-discrete approximation on an adaptive mesh $0=t_0<t_1<\cdots$ with resolution parameter $h>0$. The numerical exit time is
\[
\nu_h := \inf\{t_n\in\mathcal T_h : \bar{\boldsymbol X}^h_{t_n}\notin D\}\wedge T
\]
The exact and numerical failure indicators and probabilities are
\begin{align}
\label{eq:exact_indicator}
\Phi_T &= \mathbf 1\{\tau_D<T\}, \qquad
\Phi_{T,h} := \mathbf 1\{\nu_h<T\},
\\
\label{eq:exact_failure_probability}
p &= \mathbb E[\Phi_T], \qquad\qquad
p_h := \mathbb E[\Phi_{T,h}]
\end{align}

\subsection{Assumptions}
\label{sec:assumptions}

The following assumptions are used throughout this paper.

\begin{sassumption}[Strong exit-time approximation rate]
\label{ass:exit_time_rate}
For the nominal closed-loop dynamics and the adaptive
numerical scheme with discretization step size $h$,
\begin{equation}
\label{eq:Lq_exit_rate}
\mathbb E|\tau - \nu_h| \leq C_\tau h^\alpha
\end{equation}
for some $\alpha > 0$ and all sufficiently small $h > 0$.
Under the conditions of Section~\ref{sec:regularity},
this holds with $\alpha = 1 - \xi$ (order~1 method)
or $\alpha = 3/2 - \xi$ (order~1.5 method)
by~\cite[Theorems~2.7 and~2.10]{hoel2024adaptive}.
\end{sassumption}

\begin{sassumption}[Mesh-point strong state approximation]
\label{ass:state_rate}
Let $\bar{\boldsymbol X}^{h}$ denote a time-discrete
approximation of the state process on an adaptive mesh
$\mathcal T_h$ with $T\in\mathcal T_h$. Then,
\begin{equation}
\label{eq:mesh_point_rate}
\mathbb E\Big[\sup_{t_k \in \mathcal{T}_h}
|\boldsymbol X(t_k) - \bar{\boldsymbol X}(t_k)|\Big]
\leq C_X h^\gamma
\end{equation}
for some $\gamma > 0$. Under the conditions of
Section~\ref{sec:regularity}, this holds with
$\gamma = 1$ or $\gamma = 3/2$
by~\cite[Proposition~3.1]{hoel2024adaptive}.
\end{sassumption}

\begin{sassumption}[Cost data regularity]
\label{ass:cost}
The augmented exit penalty \(\psi_1^\eta\) is a constant (state-independent). The augmented survival penalty \(\psi_2^\eta\) admits a bounded extension to \(\mathbb R^d\) and is Lipschitz on \(\overline{D_{\bar R_D}}\). The running cost
\(\ell^\eta(\cdot,\mathbf u,t)\) is bounded on
\(\mathbb R^d\), \(C^3\) and Lipschitz in \((\boldsymbol x,t)\)
on \(D_{\bar R_D}\), all uniformly over admissible
controls and time.
\end{sassumption}

\begin{sassumption}[Terminal tube estimate]
\label{ass:anticoncentration}
There exist \(\delta_0>0\), \(C_T>0\), and \(\beta>0\) such that, with \(\inf\emptyset:=+\infty\),
\begin{equation}
\label{eq:terminal_tube}
\mathbb P\!\Big(\inf_{t\in[T-\delta,\tau]}
\operatorname{dist}(\boldsymbol X_t,\partial D)
\le\delta\Big) \le C_T\,\delta^\beta
\qquad\forall\,\delta\in(0,\delta_0]
\end{equation}
\end{sassumption}
This assumption controls the probability that the exact
diffusion enters a thin boundary layer near the horizon $T$. It is critical in enabling the transfer of
exit-time rates to indicator rates. It implies $\mathbb P(\tau_D\in[T-\delta,T]) \leq C\delta^\beta$, which controls the exit-time
distribution near $T$, but this is
not sufficient for our purposes. The terminal tube
estimate is stronger, as it is formulated at the path
level, controlling the disagreement between
exact and numerical indicator functions.
\begin{sassumption}[Classical HJB and duality]
\label{ass:classical}\label{ass:duality}
For each fixed \(\eta\ge 0\), the value function~\eqref{eq:value_function} is a classical solution of~\eqref{eq:hjb}, a minimizing feedback \(\boldsymbol\mu^*_\eta:\overline D\times[0,T]\to\mathbb R^p\) exists and is admissible, and the induced closed-loop drift is Lipschitz in \(\boldsymbol x\) on \(\overline D\), uniformly in \(t\). An optimal multiplier $\eta^*\ge 0$ exists with complementary slackness and the mapping $\eta\mapsto \mathcal S(\mathbf u^*(\cdot;\eta))$ is continuous~\cite{PfeifferTanZhou2021}.
\end{sassumption}

Assumption \ref{ass:duality} ensures that our strong approximation results for the relaxed cost directly matter to the constrained optimal control problem.
\section{DISCRETIZATION ERROR IN EXIT-EVENT ESTIMATION}
\label{sec:discretization_error}

This section defines the numerical error quantities and provides our main theorems relating exit-time approximation and exit-functional approximation. The key quantities for safety are the clipped exit time $\tau$, the failure indicator $\Phi_T$, the stopped state $\boldsymbol X_\tau$, and the augmented terminal cost~\eqref{eq:augmented_terminal}.

\subsection{Strong and weak error for exit-based quantities}

Define the strong exit-time error, strong indicator error, and weak indicator error:
\begin{align}
\label{eq:strong_exit_time_error}
e_\tau(h) &:= \mathbb E|\tau-\nu_h| \\
\label{eq:strong_indicator_error}
e_{\Phi}^{\mathrm{s}}(h) &:= \mathbb E\big|\mathbf 1\{\tau_D<T\}-\mathbf 1\{\nu_h<T\}\big| \\
\label{eq:weak_indicator_error}
e_{\Phi}^{\mathrm{w}}(h) &:= |p-p_h|
\end{align}
Since $\Phi_T,\Phi_{T,h}\in\{0,1\}$, $\ePhis(h) = \mathbb P(\mathbf 1\{\tau_D<T\}\neq \mathbf 1\{\nu_h<T\})$, and $\ePhiw(h) \le \ePhis(h)$ by Jensen's inequality, so a strong bound implies a weak bound of at least the same order.

\subsection{Clipped exit-time to indicator transfer}

A small clipped exit-time error $\mathbb E|\tau-\nu_h|$ does not guarantee a small indicator error $\ePhis(h)$, since the numerical process can exit just before $T$ while the exact process survives, misclassifying safety despite a small $|\tau-\nu_h|$. The terminal tube estimate (Assumption~\ref{ass:anticoncentration}) bounds the probability of this by limiting how often the exact process grazes $\partial D$ near $T$.

\begin{theorem}[Clipped exit-time to indicator transfer]
\label{thm:generic_indicator_rate}
Under Assumptions~\ref{ass:exit_time_rate},
\ref{ass:state_rate}, and~\ref{ass:anticoncentration}, given $\delta_0$,
for every \(0<\delta\le\delta_0\),
\begin{equation}
\label{eq:generic_indicator_three_term}
\ePhis(h)
=
\mathbb E|\Phi_T-\Phi_{T,h}|
\le
C\Big(\delta^{\brate} + \frac{h^{\arate}}{\delta}
+ \frac{h^{\grate}}{\delta}\Big)
\end{equation}
for all $h\le\delta_0^{(1+{\brate})/{\arate}}$. If, in addition, \(\grate\ge\arate\), then
\begin{equation}
\label{eq:generic_indicator_rate}
\ePhis(h),\;
\ePhiw(h)
\le
C\,h^{\frac{{\arate}{\brate}}{1+{\brate}}}
\end{equation}
\end{theorem}

\begin{proof}
Define the events
\begin{equation*}
\begin{aligned}
A_\delta &:= \Big\{\inf_{t\in[T-\delta,\tau]} \operatorname{dist}(\boldsymbol X_t,\partial D)\le\delta\Big\} \\
B_{h,\delta} &:= \Big\{\sup_{t_k\in\mathcal T_h,\,t_k\le T}\|\boldsymbol X_{t_k}-\bar{\boldsymbol X}^h_{t_k}\|>\delta\Big\}
\end{aligned}
\end{equation*}
We claim
\begin{equation}
\label{eq:proof_three_set}
\{\Phi_T\neq\Phi_{T,h}\}
\subseteq
A_\delta \cup \{|\tau-\nu_h|>\delta\} \cup B_{h,\delta}
\end{equation}
To see this: if \(\tau_D<T\) but \(\nu_h=T\), then either \(\tau_D\in[T-\delta,T)\), placing the exact path in \(A_\delta\), or \(\tau_D\le T-\delta\), so \(|\tau-\nu_h|=T-\tau_D>\delta\). If \(\tau_D\ge T\) but \(\nu_h<T\), then either \(\nu_h\le T-\delta\), so \(|\tau-\nu_h|>\delta\), or \(\nu_h\in[T-\delta,T)\). In the latter case, at the mesh time \(t_k=\nu_h\), the numerical path exits (\(\bar{\boldsymbol X}^h_{t_k}\notin D\)) while the exact path survives (\(\boldsymbol X_{t_k}\in D\)). If \(B_{h,\delta}^c\) holds, then \(\|\boldsymbol X_{t_k}-\bar{\boldsymbol X}^h_{t_k}\|\le\delta\), so \(\operatorname{dist}(\boldsymbol X_{t_k},\partial D)\le\delta\), and \(A_\delta\) occurs.

Taking probabilities in~\eqref{eq:proof_three_set} via the union bound and applying Assumption~\ref{ass:anticoncentration} and Markov's inequality to~\eqref{eq:Lq_exit_rate} and~\eqref{eq:mesh_point_rate} gives~\eqref{eq:generic_indicator_three_term}. When \(\grate\ge\arate\), the \(h^{\grate}/\delta\) term is dominated by \(h^{\arate}/\delta\), and optimizing \(\delta=h^{{\arate}/(1+{\brate})}\) gives~\eqref{eq:generic_indicator_rate}. The weak bound follows from $\ePhiw(h)\le\ePhis(h)$.
\end{proof}

\begin{remark}
When the terminal region is well separated from $\partial D$, the terminal tube estimate \ref{ass:anticoncentration} holds for all $\brate>0$ for $\delta$ above a threshold, and the indicator rate recovers the full exit-time rate $O(h^{\arate})$ as $\brate\to\infty$. More generally, $\brate$ quantifies clearance from $\partial D$ near the horizon, with $\brate=1$ the standard bounded density regime and $\brate>1$ stronger decay.
\end{remark}

\subsection{Running-cost approximation}
\label{subsec:running_cost_approximation}
Fix a feedback law \(\boldsymbol\mu\) and define \(Y_t:=\int_0^t \ell^{\eta,\mu}(\boldsymbol X_s,s)\,ds\) with \(\ell^{\eta,\mu}(\boldsymbol x,t):=\ell^\eta(\boldsymbol x,\boldsymbol\mu(\boldsymbol x,t),t)\). By augmenting the state to \(\boldsymbol Z_t:=(\boldsymbol X_t,t,Y_t)\in\mathbb R^{d+2}\) and applying the strong It\^o--Taylor mesh-point theorem~\cite[Theorem~10.6.3]{kloeden1992numerical} to the enlarged system (whose adaptive mesh is determined by the \(\boldsymbol X\)-component alone), the running-cost accumulator satisfies
\begin{equation}
\label{eq:augmented_mesh_rate}
\mathbb E\Big[\sup_{t_k\in\mathcal T_h}|Y_{t_k}-\bar Y^h_{t_k}|\Big]
\le C h^{\grate}
\end{equation}
For \(\grate=3/2\), this requires the order-matched It\^o--Taylor update for \(Y\). Assuming a truncated adaptive mesh (\(T\in\mathcal T_h\)), the stopped running-cost error satisfies
\begin{equation}
\label{eq:running_cost_rate}
\mathbb E|Y_\tau - \bar Y^h(\nu_h)|
\le
C\,h^{\min({\grate},{\arate})}
\end{equation}
by splitting \(|Y_\tau - \bar Y^h(\nu_h)|\le |Y_\tau - Y_{\nu_h}| + |Y_{\nu_h} - \bar Y^h(\nu_h)|\) and applying~\eqref{eq:Lq_exit_rate} and~\eqref{eq:augmented_mesh_rate}, with $\bar Y^h$ clipped at $\pm\|\ell^\eta\|_\infty T$.

\subsection{Reduction theorem for the relaxed stopped cost}

We now combine the indicator bound (Theorem~\ref{thm:generic_indicator_rate}) and the running-cost bound~\eqref{eq:running_cost_rate} into a single theorem for the relaxed stopped cost. The augmented exit penalty \(\psi_1^\eta\) is state-independent (Assumption~\ref{ass:cost}), while the augmented survival penalty \(\psi_2^\eta\) is evaluated at the deterministic terminal time \(T\) using the bounded extension from Assumption~\ref{ass:cost}, and is Lipschitz on \(\overline{D_{\bar R_D}}\). Recall the exact and numerical relaxed stopped costs
\begin{align}
\label{eq:stopped_cost_exact_compact}
\mathcal{L}_\eta
&:=
\psi_1^\eta\,\mathbf 1\{\tau_D<T\}
+
\psi_2^\eta(\boldsymbol X_T)\,\mathbf 1\{\tau_D\ge T\}
\nonumber\\
&\quad+
Y_\tau
-\eta\Delta
\end{align}
The numerical counterpart $\Setah$\label{eq:stopped_cost_num_compact} is defined analogously with $(\tau_D,\boldsymbol X_T,Y_\tau)$ replaced by $(\nu_h,\bar{\boldsymbol X}^h_T,\bar Y^h(\nu_h))$.

\begin{theorem}[Stopped-functional reduction]
\label{thm:stopped_functional}
Under Assumptions~\ref{ass:exit_time_rate}--\ref{ass:anticoncentration}, suppose~\eqref{eq:augmented_mesh_rate} holds and the adaptive mesh is truncated so that \(T\in\mathcal T_h\). Then there exists \(C>0\) such that for all \(h\) sufficiently small
\begin{align}
\label{eq:stopped_cost_rate_compact}
\mathbb E|\Seta - \Setah|
&\le
C_{\mathrm{sw}}\,
\mathbb E\big|\mathbf 1\{\tau_D<T\}-\mathbf 1\{\nu_h<T\}\big|
\nonumber\\
&\quad+
C\,h^{\min({\grate},{\arate})}
\end{align}
where
\[
C_{\mathrm{sw}}:=|\psi_1^\eta|+\|\psi_2^\eta\|_{L^\infty(\mathbb R^d)}
\]
The indicator term is controlled by Theorem~\ref{thm:generic_indicator_rate}.
\end{theorem}

\begin{proof}
Decompose the survival term without factoring \(\psi_2^\eta(\boldsymbol X_T)\) across all paths:
\begin{align}
\label{eq:proof_phi2_decomp}
&\psi_2^\eta(\boldsymbol X_T)\mathbf 1\{\tau_D\ge T\}
-
\psi_2^\eta(\bar{\boldsymbol X}^h_T)\mathbf 1\{\nu_h\ge T\}
\nonumber\\
&\quad=
\big(\psi_2^\eta(\boldsymbol X_T)-\psi_2^\eta(\bar{\boldsymbol X}^h_T)\big)
\mathbf 1\{\tau_D\ge T,\,\nu_h\ge T\}
\nonumber\\
&\qquad+
\psi_2^\eta(\boldsymbol X_T)\mathbf 1\{\tau_D\ge T,\,\nu_h<T\}
\nonumber\\
&\qquad-
\psi_2^\eta(\bar{\boldsymbol X}^h_T)\mathbf 1\{\tau_D<T,\,\nu_h\ge T\}
\end{align}
Taking absolute values and splitting according to whether
\(\|\boldsymbol X_T-\bar{\boldsymbol X}^h_T\|\le \bar R_D\),
\begin{align}
\label{eq:proof_reduction}
|\Seta - \Setah|
&\le
L_{\psi_2^\eta}\|\boldsymbol X_T - \bar{\boldsymbol X}^h_T\|\,
\mathbf 1\{E_{\mathrm{Lip}}\}
\nonumber\\
&\quad+
2\|\psi_2^\eta\|_{L^\infty(\mathbb R^d)}\,
\mathbf 1\{E_{\mathrm{Lip}}^c\cap\{\tau_D\ge T,\,\nu_h\ge T\}\}
\nonumber\\
&\quad+
C_{\mathrm{sw}}\,
\big|\mathbf 1\{\tau_D<T\}-\mathbf 1\{\nu_h<T\}\big|
\nonumber\\
&\quad+
|Y_\tau - \bar Y^h(\nu_h)|
\end{align}
where \(E_{\mathrm{Lip}}:=\{\tau_D\ge T,\,\nu_h\ge T,\,\|\boldsymbol X_T-\bar{\boldsymbol X}^h_T\|\le \bar R_D\}\). On \(E_{\mathrm{Lip}}\), \(\boldsymbol X_T\in\overline D\) and \(\bar{\boldsymbol X}^h_T\in\overline{D_{\bar R_D}}\), so the Lipschitz bound for \(\psi_2^\eta\) applies. Since \(T\in\mathcal T_h\), the mesh-point strong estimate gives
\[
\mathbb E\|\boldsymbol X_T-\bar{\boldsymbol X}^h_T\|\le C h^{\grate}
\]
and by Markov's inequality,
\[
\mathbb P\!\big(\|\boldsymbol X_T-\bar{\boldsymbol X}^h_T\|>\bar R_D\big)
\le
\bar R_D^{-1}\,C h^{\grate}
\]
Therefore the survival-term contribution is \(O(h^{\grate})\). Applying this together with~\eqref{eq:running_cost_rate} and Theorem~\ref{thm:generic_indicator_rate} gives~\eqref{eq:stopped_cost_rate_compact}.
\end{proof}

\begin{corollary}
\label{cor:leading_rate}
Under the assumptions of Theorem~\ref{thm:stopped_functional},
\begin{equation}
\label{eq:leading_rate}
\mathbb E|\Seta - \Setah|
\le
C\,h^{\krate}
\end{equation}
with \(\krate:=\min\!\big\{\frac{{\arate}{\brate}}{1+{\brate}},\,\min({\grate},{\arate})\big\}\). When the state and running-cost rates are order-matched so that \(\min({\grate},{\arate})={\arate}\) as in Section~\ref{sec:higher_order_methods}, the indicator term is the dominant term.
\end{corollary}
In practice, a controller using the numerical relaxed stopped cost $\Setah$ in place of $\Seta$ incurs an error that vanishes at rate $O(h^{\krate})$. Improving the exit-time propagator (larger $\arate$), therefore, directly tightens the cost estimate and, through dual ascent, the constraint satisfaction of the resulting controller.
\section{HIGHER-ORDER ADAPTIVE METHODS}
\label{sec:higher_order_methods}
This section bridges the exit-time approximation rates \(\alpha\) of \cite{hoel2024adaptive} to the Theorems of Section~\ref{sec:discretization_error}. Two cases are studied, distinguished by the commutativity structure of the diffusion. Writing the diffusion coefficient in~\eqref{eq:controlled_sde} columnwise as $\mathbf B=[\mathbf b_1,\ldots,\mathbf b_m]$, define $L_j := \sum_{k=1}^d b_j^{(k)}(\mathbf x)\,\partial_{x_k}$. The commutativity conditions are
\begin{align}
\label{eq:first_commutativity}
L_i\mathbf b_j &= L_j\mathbf b_i
\quad \forall\, i,j \\
\label{eq:second_commutativity}
L_iL_j\mathbf b_k &= L_jL_i\mathbf b_k
\quad \forall\, i,j,k
\end{align}

If~\eqref{eq:first_commutativity}--\eqref{eq:second_commutativity} are not imposed, the natural higher-order choice is the adaptive order-\(1\) Milstein method with explicit approximation of the off-diagonal It\^o integrals. If they do hold, then the adaptive order-\(1.5\) method of~\cite{hoel2024adaptive} is available and yields a stronger exit-time rate.

\subsection{Noncommutative adaptive order-\(1\) Milstein with L\'evy areas}
\label{subsec:noncomm_order1}

For general diffusions, the commutativity conditions~\eqref{eq:first_commutativity}--\eqref{eq:second_commutativity} need not hold. In this case, we use the adaptive order-$1$ method of~\cite{hoel2024adaptive}, where the step-size is adaptively chosen from $\{h, h^2\}$ based on the distance from boundary, with explicit approximation of the off-diagonal iterated It\^o integrals via Wiktorsson's method~\cite{wiktorsson2001joint}. The Milstein discretization is used:
\begin{align}
\label{eq:noncomm_milstein_step}
\bar{\boldsymbol X}_{n+1}
&=
\bar{\boldsymbol X}_n
+
\mathbf F^{\mu}(\bar{\boldsymbol X}_n)\,\Delta t_n
+
\sum_{j=1}^m \mathbf b_j(\bar{\boldsymbol X}_n,t_n)\,\Delta W_n^j
\nonumber\\
&\quad
+
\sum_{i,j=1}^m
L_i\mathbf b_j(\bar{\boldsymbol X}_n,t_n)\,\widehat I_n^{ij}
\end{align}
where $\Delta W_n^j := W_{t_{n+1}}^j-W_{t_n}^j$. The iterated integrals are
\begin{equation}
\label{eq:offdiag_milstein_terms}
\widehat I_n^{ii}
=
\tfrac{1}{2}\big((\Delta W_n^i)^2-\Delta t_n\big),\;
\widehat I_n^{ij}
=
\tfrac{1}{2}\Delta W_n^i\Delta W_n^j + \widehat A_{n,p_n}^{ij}
\end{equation}
where $i\neq j$, $\widehat A_{n,p_n}^{ij}$ is the Wiktorsson L\'evy-area approximation~\cite{wiktorsson2001joint} with truncation level
\begin{equation}
\label{eq:wiktorsson_scaling}
p_n \asymp (\Delta t_n)^{-1/2}
\end{equation}
chosen so that the approximation error matches the local Milstein order~\cite{kastner2022analysis}.

\begin{theorem}
\label{thm:noncomm_order1_exit}
Let \(\nu_h^{\mathrm{NC}}\) be the numerical exit time produced by the adaptive Milstein scheme~\eqref{eq:noncomm_milstein_step} with~\eqref{eq:offdiag_milstein_terms} and~\eqref{eq:wiktorsson_scaling}. Then $\mathbb E|\tau-\nu_h^{\mathrm{NC}}| \le C_\xi h^{1-\xi}$\label{eq:noncomm_exit_rate} for every $\xi>0$ and all sufficiently small $h$.
\end{theorem}

\begin{proof}
The proof follows~\cite[Theorem~2.7]{hoel2024adaptive} with exact off-diagonal It\^o integrals replaced by Wiktorsson approximations. By the truncation analysis of~\cite{kastner2022analysis}, this preserves the local $L^2$ order, and the remainder of the proof applies verbatim.
\end{proof}

\begin{corollary}
\label{cor:noncomm_indicator_rate}
Theorem~\ref{thm:generic_indicator_rate} with $\arate=1-\xi$ gives
$\ePhis(h) \le C\,h^{(1-\xi){\brate}/(1+{\brate})}$.
\end{corollary}

For any $\brate>0$, this roughly doubles the indicator exponent compared to EM ($\arate\brate/(1+\brate)$ vs.\ $\arate\brate/(2(1+\brate))$). As $\brate\to\infty$, the rate approaches the full exit-time rate $1-\xi$.

\subsection{Commutative adaptive order-\(1.5\) method}
\label{subsec:comm_order15}

Under~\eqref{eq:first_commutativity}--\eqref{eq:second_commutativity}, the adaptive order-\(1.5\) method of~\cite{hoel2024adaptive} uses the three-scale mesh \(\{h,h^2,h^3\}\) based on the boundary distance. Let \(\nu_h^{\mathrm{C}}\) denote the resulting numerical exit time.

For the running-cost component, the order-matched It\^o--Taylor update for $Y$ at order $1.5$ uses the same higher-order random increments as the state update, as mentioned in Section ~\ref{subsec:running_cost_approximation}, since a Riemann sum retains only order~$1$.

\begin{theorem}[{\cite[Theorem 2.10]{hoel2024adaptive}}]
\label{thm:comm_order15_exit}
Under the assumptions of the adaptive commutative order-\(1.5\) method in~\cite{hoel2024adaptive}, $\mathbb E|\tau-\nu_h^{\mathrm{C}}| \le C_\xi h^{3/2-\xi}$\label{eq:comm_exit_rate} for every $\xi>0$ and all sufficiently small $h$.
\end{theorem}

\begin{corollary}
\label{cor:comm_indicator_rate}
Theorem~\ref{thm:generic_indicator_rate} with $\arate=3/2-\xi$ gives
$\ePhis(h) \le C\,h^{(3/2-\xi){\brate}/(1+{\brate})}$.
\end{corollary}

Already at $\brate=1$ (bounded density), this gives order near $3/4$, a substantial gain over EM's $1/4$. When the commutativity conditions hold, the order-$1.5$ rate improves on the order-$1$ rate for any $\brate$, without requiring L\'evy-area simulation. 

The stopped-cost rates follow from Theorem~\ref{thm:stopped_functional} and coincide with the indicator rates.
\section{STRONG APPROXIMATION IN CHANCE-CONSTRAINED PATH-INTEGRAL CONTROL}
\label{sec:ccpic}

We now specialize the framework of Section~\ref{sec:problem} to control-affine dynamics with control-quadratic running cost.

\subsection{Path-integral-solvable subclass}

Restrict~\eqref{eq:controlled_sde} to control-affine drift, control-quadratic running cost, and the matching condition,
\begin{subequations}
\begin{align}
\label{eq:control_affine}
\mathbf F(\boldsymbol x,\mathbf u)
&=
\mathbf f(\boldsymbol x) + \mathbf G(\boldsymbol x)\,\mathbf u \\
\label{eq:running_cost}
\ell(\boldsymbol x,\mathbf u,t)
&=
\tfrac{1}{2}\,\mathbf u^\top \mathbf R(\boldsymbol x)\,\mathbf u
+ q(\boldsymbol x,t) \\
\label{eq:matching}
\mathbf B\mathbf B^\top
&=
\lambda\,\mathbf G\mathbf R^{-1}\mathbf G^\top
\end{align}
\end{subequations}
with \(\mathbf R\succ 0\) and \(\lambda>0\). The Cole--Hopf transform \(\xi=\exp(-V/\lambda)\) linearizes the HJB equation~\eqref{eq:hjb}. The Feynman--Kac lemma gives $\xi$ as an expectation over the uncontrolled process $\hat{\boldsymbol X}_t$ (i.e.,~\eqref{eq:controlled_sde} with $\mathbf u\equiv\mathbf 0$),
\begin{equation}
\label{eq:feynman_kac}
\xi(\boldsymbol x,t;\eta)
=
\mathbb E\Big[
\exp\!\Big(
\!{-}\tfrac{1}{\lambda}\big[\psi^\eta
+\!\textstyle\int_t^{\hat\tau}\! q(\hat{\boldsymbol X}_s,s)\,ds\big]
\Big)\Big]
\end{equation}
where \(\hat\tau := \inf\{s\ge t : \hat{\boldsymbol X}_s\notin D\}\wedge T\). The optimal control is recovered by differentiating~\eqref{eq:feynman_kac}~\cite{theodorou2010generalized,patil2025strong}.

For the chance-constrained specialization of Section~\ref{sec:problem}, the augmented terminal cost, \(\varphi^\eta\) in this case, is given by~\eqref{eq:augmented_terminal} with $\Gamma_1=1, \Gamma_2=0$, and the dual ascent update for \(\eta\) is
\begin{equation}
\label{eq:dual_ascent}
\eta \leftarrow \eta + \gamma_\eta\big(p(\mathbf u^*(\cdot;\eta))-\Delta\big)
\end{equation}
where the failure probability \(p(\mathbf u^*(\cdot;\eta))\) is estimated from the same sampled trajectories used for the control update. Under Assumption~\ref{ass:duality}, strong duality holds and the optimal multiplier \(\eta^*\) satisfying complementary slackness exists~\cite{patil2025strong}.

\subsection{Approximation of the closed-loop failure probability}

Under the path-integral representation~\eqref{eq:feynman_kac}, the failure probability of the optimal control for fixed \(\eta\) is computed from uncontrolled rollouts via importance weighting. Define the exact and numerical exponential weights \cite{patil2025strong}
\begin{align}
\label{eq:weight_exact}
R_\eta &:= \exp(-\Seta/\lambda) \\
\label{eq:weight_numerical}
R_{\eta,h} &:= \exp(-\Setah/\lambda)
\end{align}
The failure probability under the optimal control \(\mathbf u^*(\cdot;\eta)\) is
\begin{equation}
\label{eq:closed_loop_pf}
p(\mathbf u^*(\cdot;\eta))
=
\frac{\mathbb E[R_\eta\,\Phi_T]}{\mathbb E[R_\eta]}
\end{equation}
The numerical estimate $p_h(\mathbf u^*(\cdot;\eta))$\label{eq:closed_loop_pf_numerical} is defined analogously with $(R_\eta,\Phi_T)$ replaced by $(R_{\eta,h},\Phi_{T,h})$.

\begin{figure*}[b]
\centering
\begin{subfigure}[b]{0.24\textwidth}
\includegraphics[width=\textwidth]{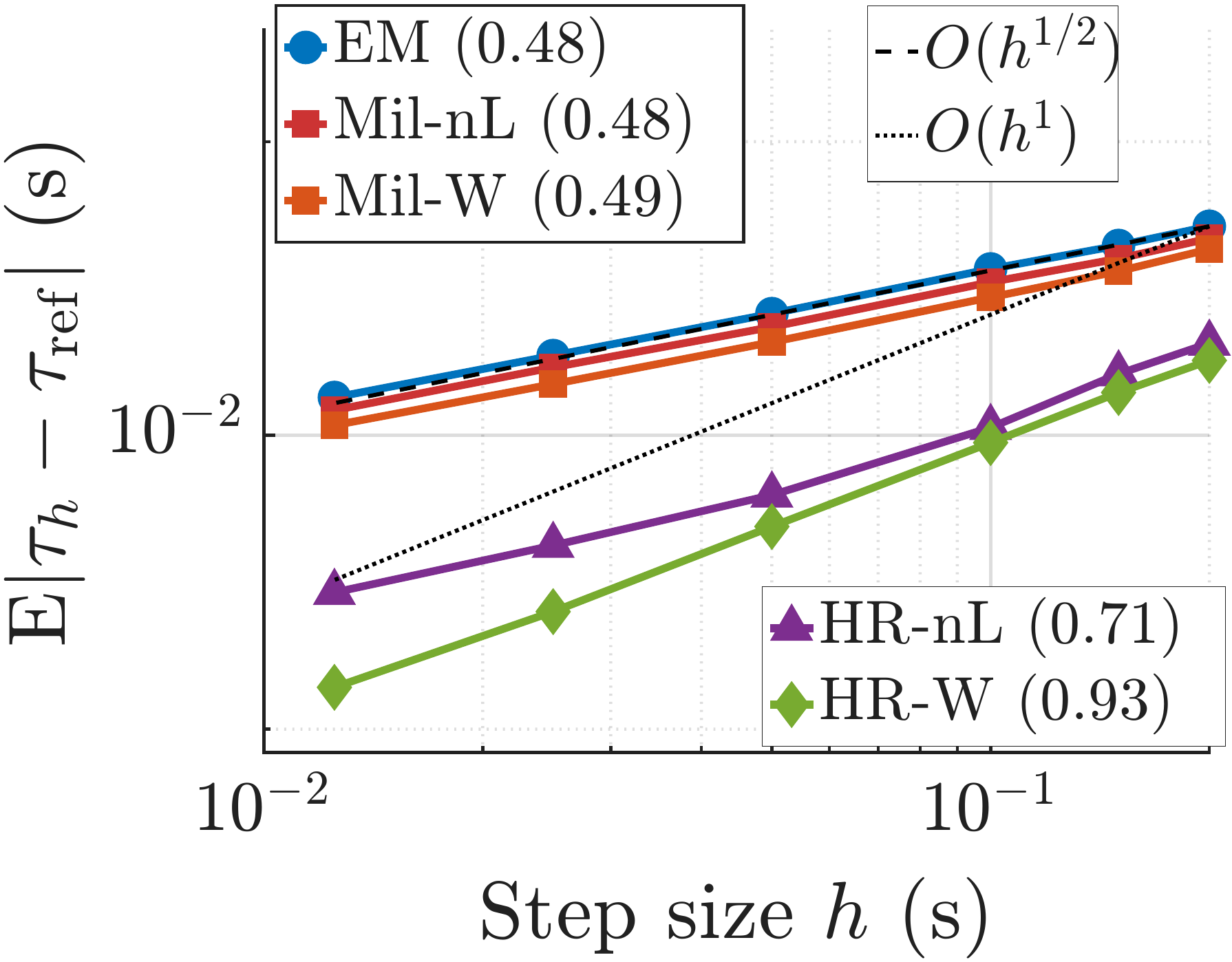}
\caption{Noncomm.: exit-time}
\label{fig:caseA_exit_time}
\end{subfigure}
\hfill
\begin{subfigure}[b]{0.24\textwidth}
\includegraphics[width=\textwidth]{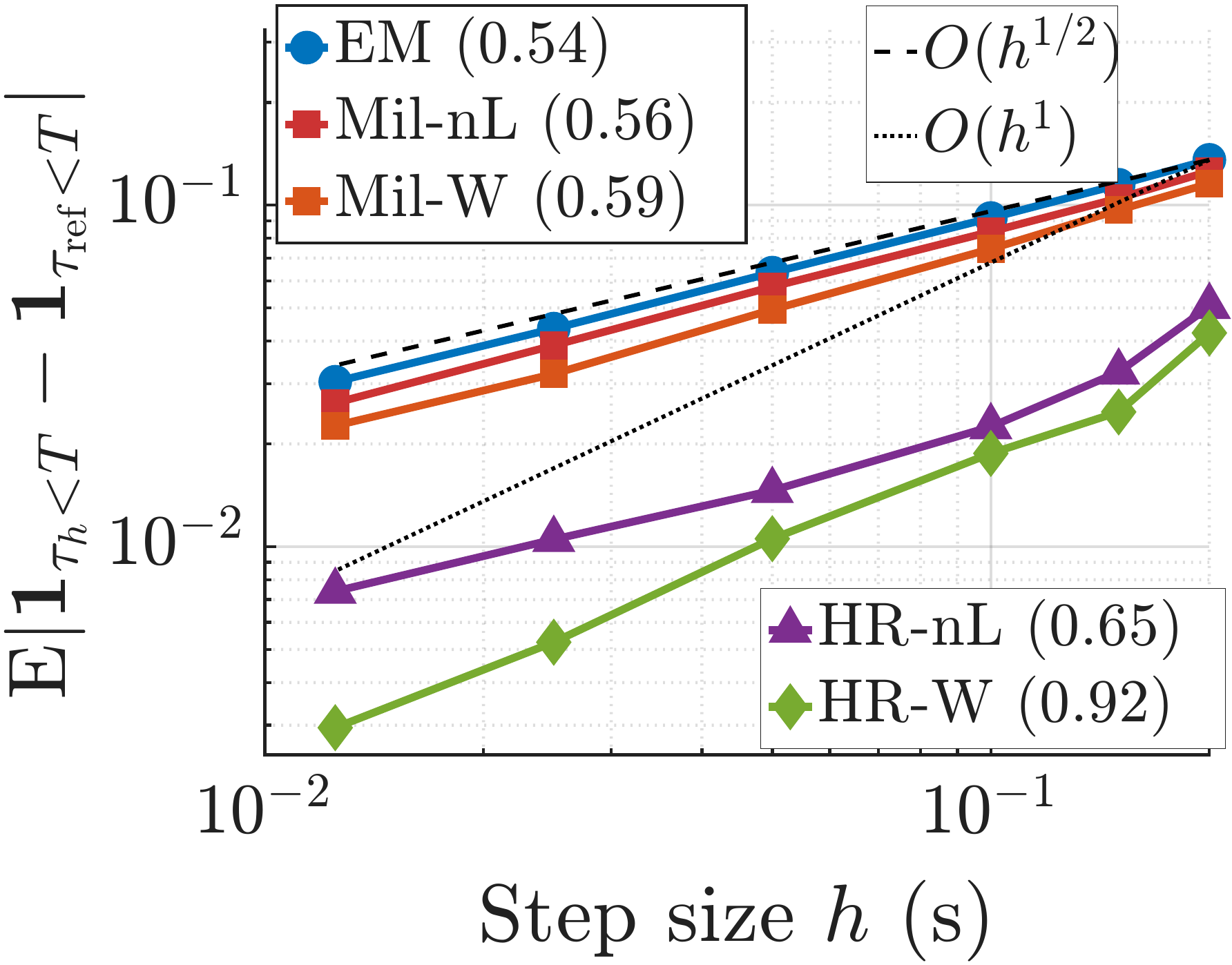}
\caption{Noncomm.: indicator}
\label{fig:caseA_indicator}
\end{subfigure}
\hfill
\begin{subfigure}[b]{0.24\textwidth}
\includegraphics[width=\textwidth]{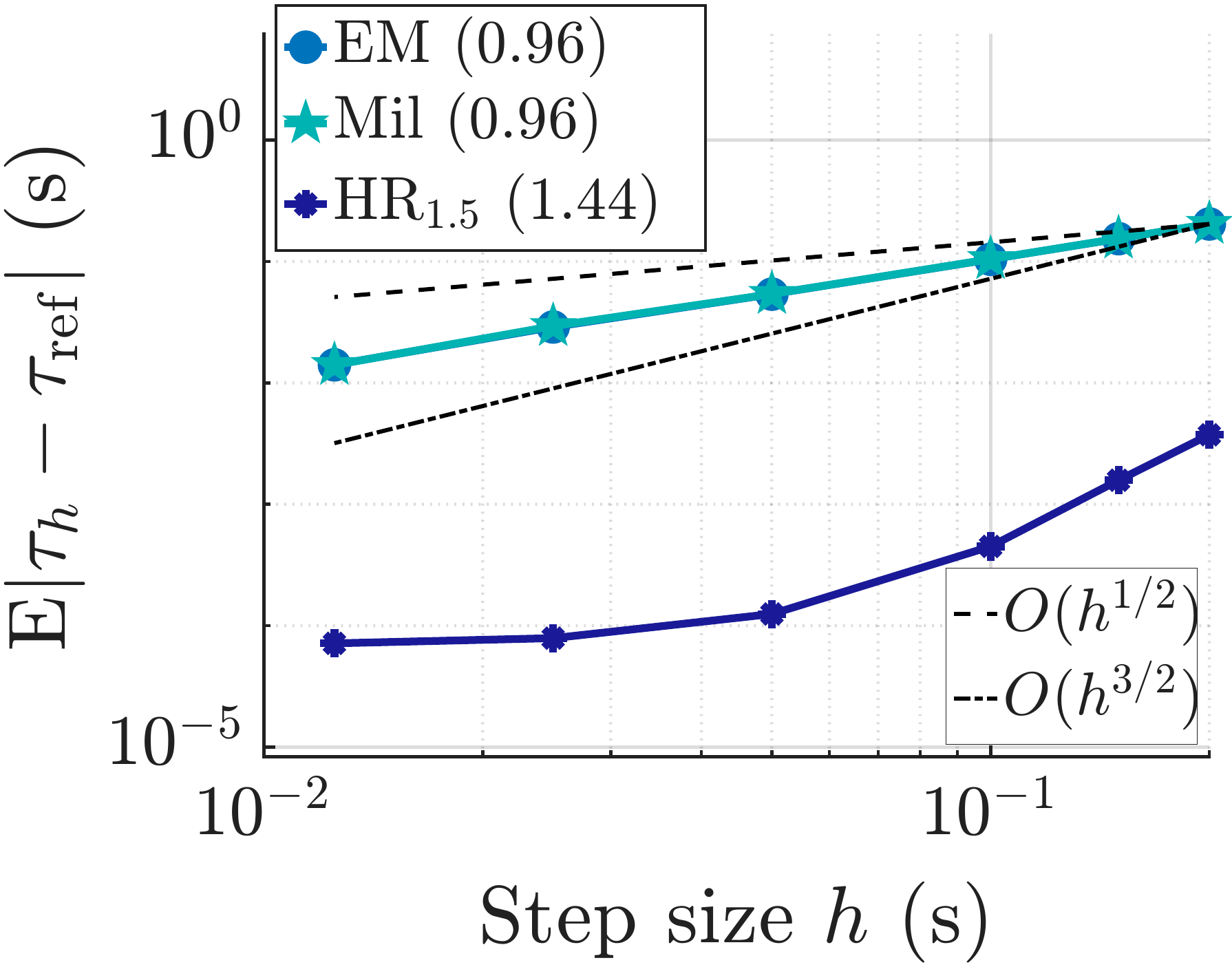}
\caption{Comm.: exit-time}
\label{fig:caseB_exit_time}
\end{subfigure}
\hfill
\begin{subfigure}[b]{0.24\textwidth}
\includegraphics[width=\textwidth]{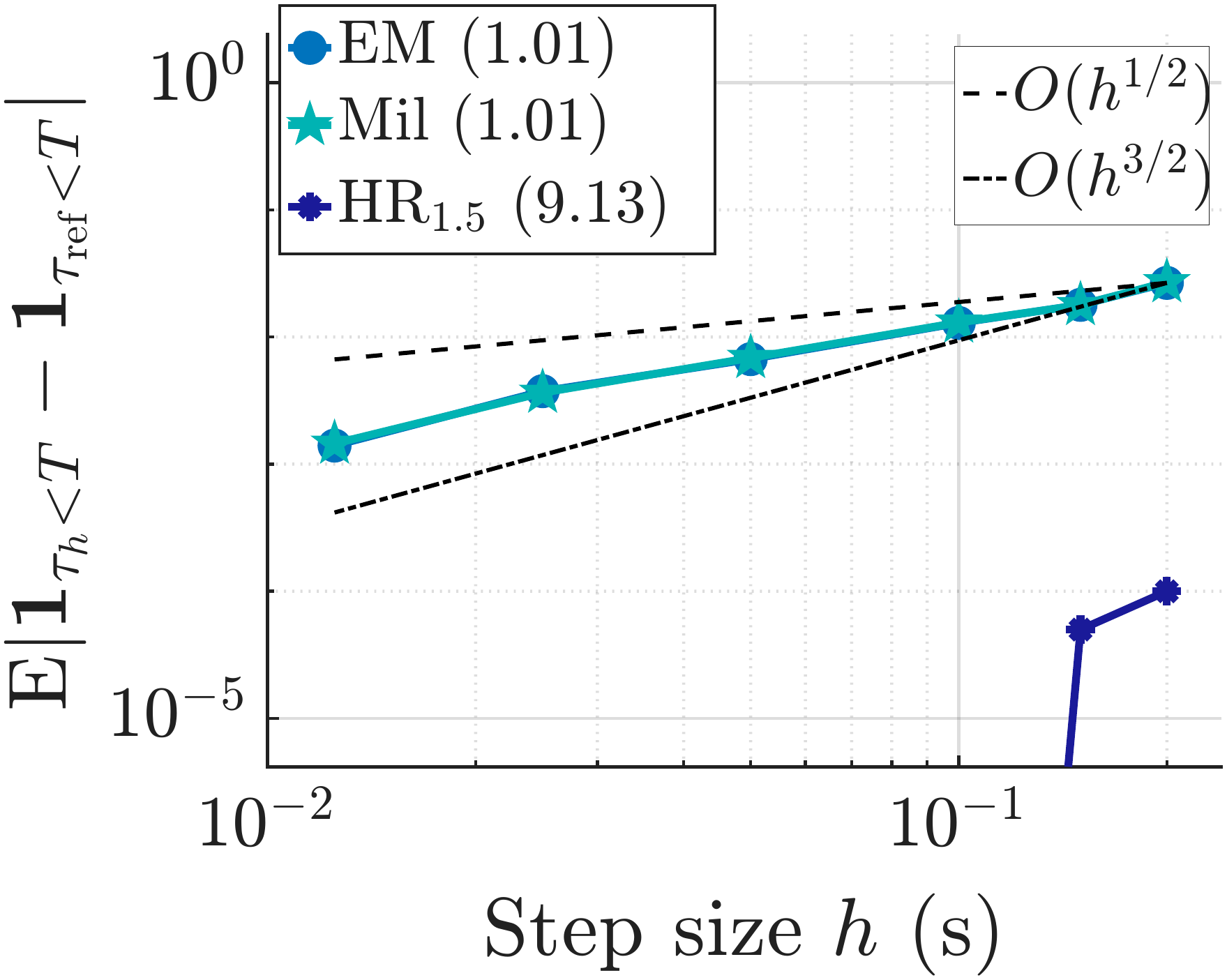}
\caption{Comm.: indicator}
\label{fig:caseB_indicator}
\end{subfigure}
\caption{Strong errors versus step size $h$. Fitted slopes reported in legends. Adaptive higher-order methods achieve best rates.}
\label{fig:strong_errors}
\end{figure*}

\begin{figure*}[b]
\centering
\begin{subfigure}[b]{0.24\textwidth}
\includegraphics[width=\textwidth]{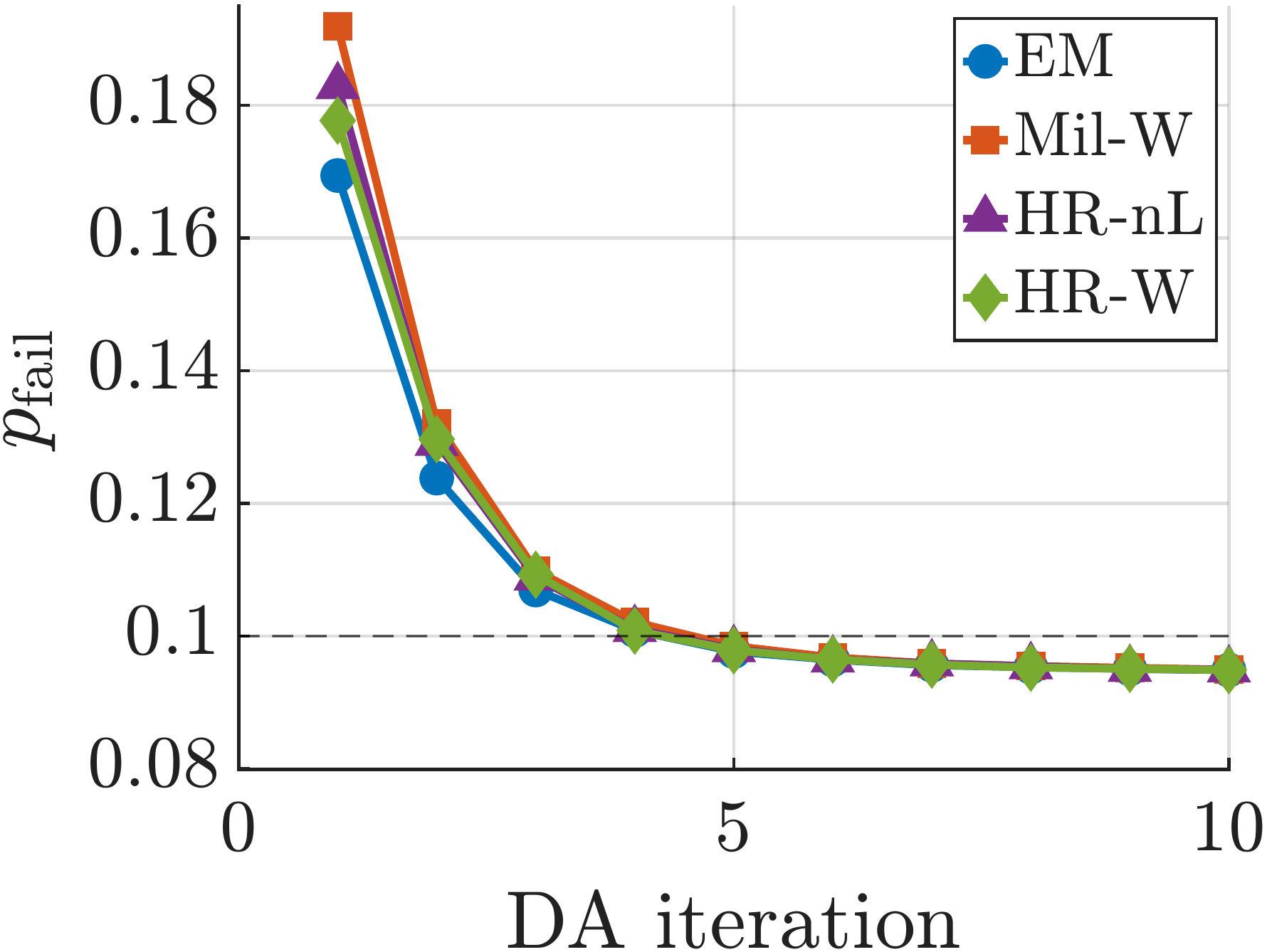}
\caption{Noncomm.: dual ascent}
\label{fig:caseC_da_noncomm}
\end{subfigure}
\hfill
\begin{subfigure}[b]{0.24\textwidth}
\includegraphics[width=\textwidth]{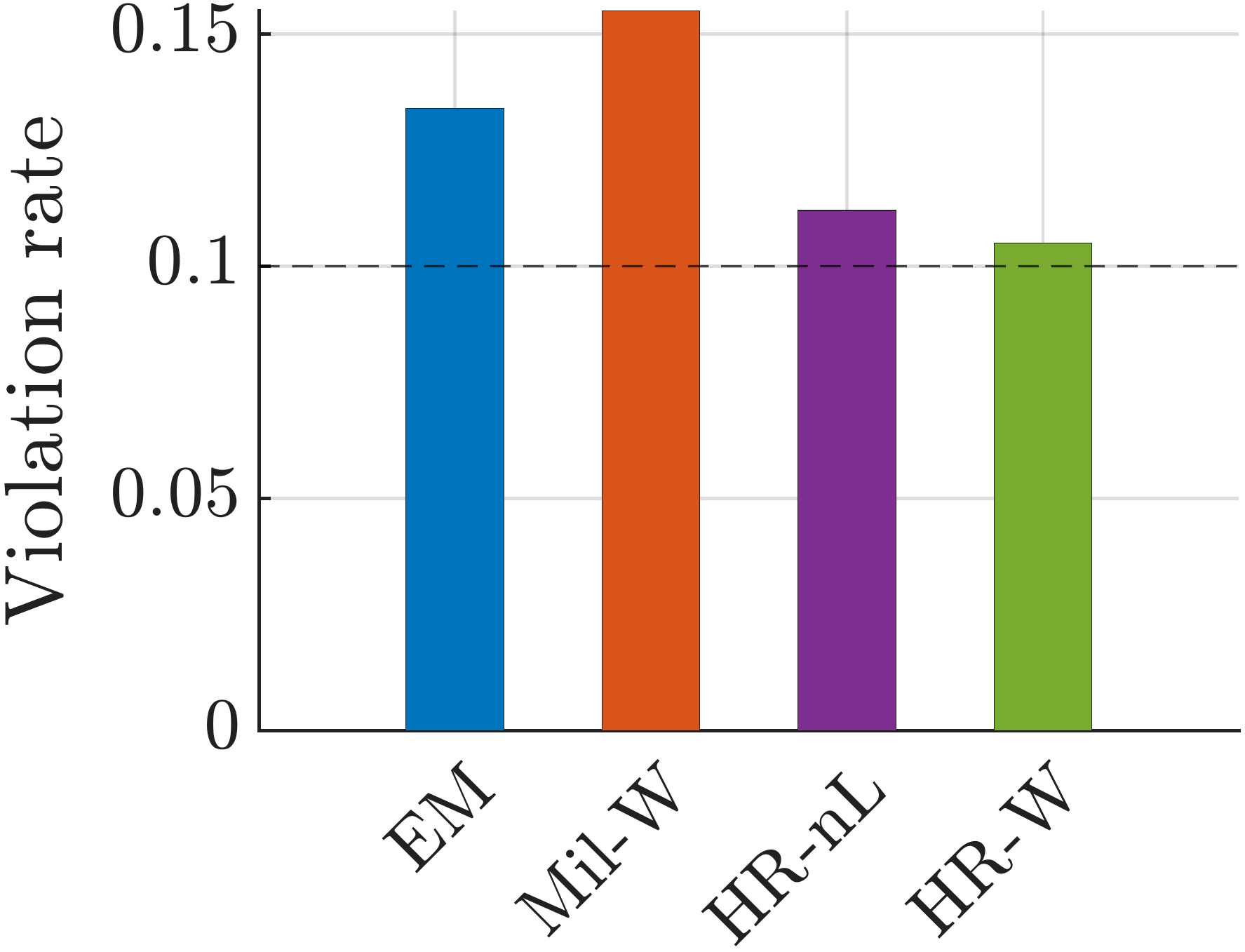}
\caption{Noncomm.: violation}
\label{fig:caseC_viol_noncomm}
\end{subfigure}
\hfill
\begin{subfigure}[b]{0.24\textwidth}
\includegraphics[width=\textwidth]{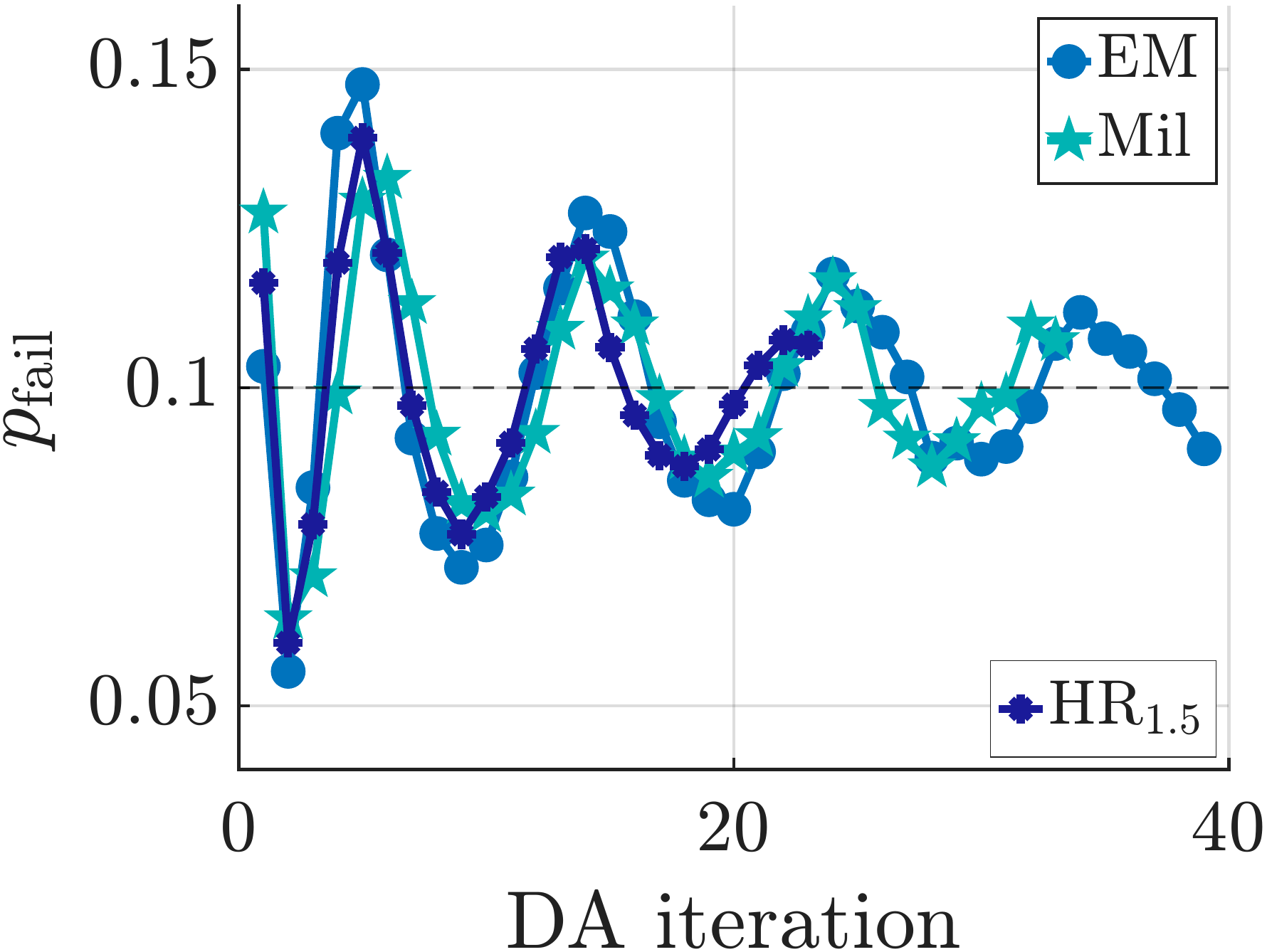}
\caption{Comm.: dual ascent}
\label{fig:caseC_da_comm}
\end{subfigure}
\hfill
\begin{subfigure}[b]{0.24\textwidth}
\includegraphics[width=\textwidth]{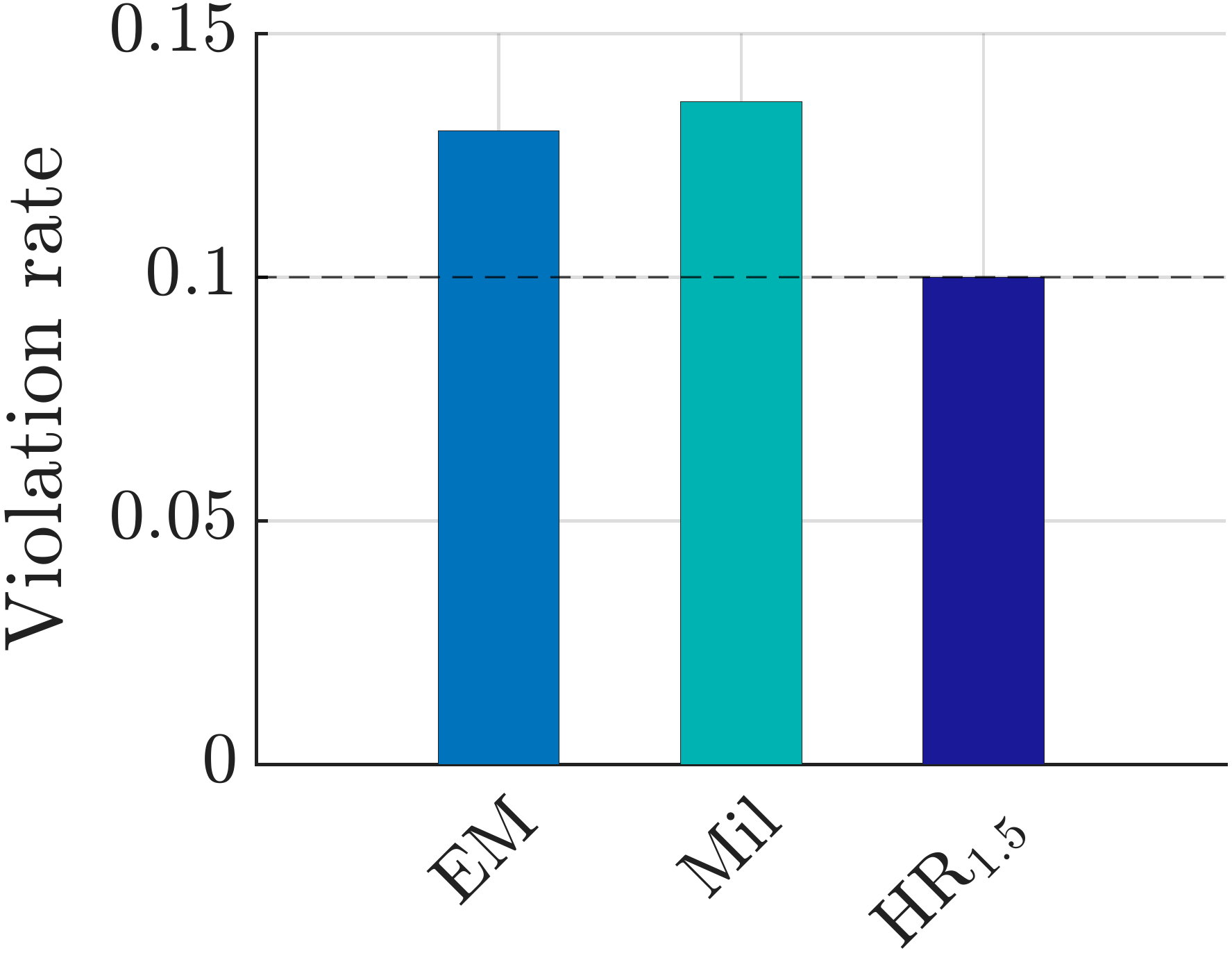}
\caption{Comm.: violation}
\label{fig:caseC_viol_comm}
\end{subfigure}
\caption{Dual-ascent convergence and empirical violation rates. Proposed methods achieve tighter constraint satisfaction at $\Delta=0.10$.}
\label{fig:caseC_da_viol}
\end{figure*}

\begin{corollary}[Closed-loop failure probability approximation]
\label{cor:closed_loop_pf}
Under Assumptions~\ref{ass:exit_time_rate}--\ref{ass:duality} and the matching condition~\eqref{eq:matching}, suppose \(|\Seta|\vee|\Setah|\le M\) almost surely and \(\mathbb E[R_\eta]\ge m>0\), \(\mathbb E[R_{\eta,h}]\ge m>0\), with \(m=e^{-M/\lambda}\) and \(\grate\ge\arate\). Then,
\begin{alignat}{2}
\label{eq:weight_rate}
&\mathbb E|R_\eta - R_{\eta,h}|
\le
\tfrac{e^{M/\lambda}}{\lambda}\,\mathbb E|\Seta - \Setah|
&\le\;
C\,h^{\krate} \\
\label{eq:weighted_indicator_rate}
&\mathbb E\big|R_\eta\,\Phi_T - R_{\eta,h}\,\Phi_{T,h}\big|
&\le\;
C\,h^{\krate}
\end{alignat}
Consequently, $|p(\mathbf u^*) - p_h(\mathbf u^*)| \le C\,h^{\krate}$.
\end{corollary}

\begin{proof}
For~\eqref{eq:weight_rate}, the exponential is Lipschitz on \([-M,M]\) with constant \(e^{M/\lambda}/\lambda\), so
\[
|R_\eta - R_{\eta,h}|
\le
\frac{e^{M/\lambda}}{\lambda}\,|\Seta - \Setah|
\]
and Theorem~\ref{thm:stopped_functional} gives the rate. For~\eqref{eq:weighted_indicator_rate}, decompose
\[
|R_\eta\,\Phi_T - R_{\eta,h}\,\Phi_{T,h}|
\le
|R_\eta - R_{\eta,h}|
+
|R_{\eta,h}|\,|\Phi_T - \Phi_{T,h}|
\]
and apply~\eqref{eq:weight_rate} and Theorem~\ref{thm:generic_indicator_rate}. Finally, apply the inequality
\[
\bigg|\frac{A}{B}-\frac{\widehat A}{\widehat B}\bigg|
\le
\frac{|A-\widehat A|}{|B|}
+
\frac{|\widehat A|\,|B-\widehat B|}{|B|\,|\widehat B|}
\]
with \(A=\mathbb E[R_\eta\Phi_T]\), \(B=\mathbb E[R_\eta]\), and \(B,\widehat B\ge m\).
\end{proof}

\begin{remark}
Section~\ref{sec:regularity} requires uniform ellipticity and a $C^4$ boundary. In practice, many systems of interest are hypoelliptic (e.g., unicycles) and operate in domains whose boundaries are not $C^4$. In Section~\ref{sec:results}, we test whether the predicted rate improvements persist in practical settings.
\end{remark}

\subsection{Controller algorithm}

Algorithm~\ref{alg:ccpic} states the controller loop. The higher-order propagator \(\mathcal{P}\) on line~5 is the only component that differs between the two cases, with the running-cost integration matched to the order as discussed in Section~\ref{sec:higher_order_methods}.
\vspace{1em}
\begin{algorithm}[H]
\caption{CC-PIC with higher-order propagation}
\label{alg:ccpic}
\begin{algorithmic}[1]
\REQUIRE Initial state $\mathbf{x}_0$, horizon $T$, rollout budget $K$, tolerance $\Delta$, step sizes, $\gamma_\eta,\lambda$, propagator $\mathcal{P}$
\STATE Initialize multiplier $\eta\ge 0$
\REPEAT
\FOR{$k=1,\dots,K$}
  \STATE Sample noise $\{\boldsymbol\epsilon_n^{(k)}\}$
  \STATE Propagate $\hat{\boldsymbol X}^{(k)}$ using $\mathcal{P}$ (uncontrolled, $\mathbf u\equiv\mathbf 0$)
  \STATE Record $\nu_h^{(k)}$, $\Phi_{T,h}^{(k)}$, and $\Setah^{(k)}$ via~\eqref{eq:stopped_cost_num_compact}
\ENDFOR
\STATE $w^{(k)}\!\leftarrow\!\exp(-\Setah^{(k)}/\lambda)\big/\sum_j\exp(-\Setah^{(j)}/\lambda)$
\STATE Recover control update from weighted samples~\cite{theodorou2010generalized}
\STATE $\eta\leftarrow\eta+\gamma_\eta\big(\sum_k w^{(k)}\Phi_{T,h}^{(k)}-\Delta\big)$
\UNTIL{convergence or receding-horizon shift}
\end{algorithmic}
\end{algorithm}

\section{RESULTS}
\label{sec:results}
\enlargethispage{-\baselineskip}

\begin{figure*}[t]
\centering
\begin{subfigure}[b]{0.24\textwidth}
\includegraphics[width=\textwidth]{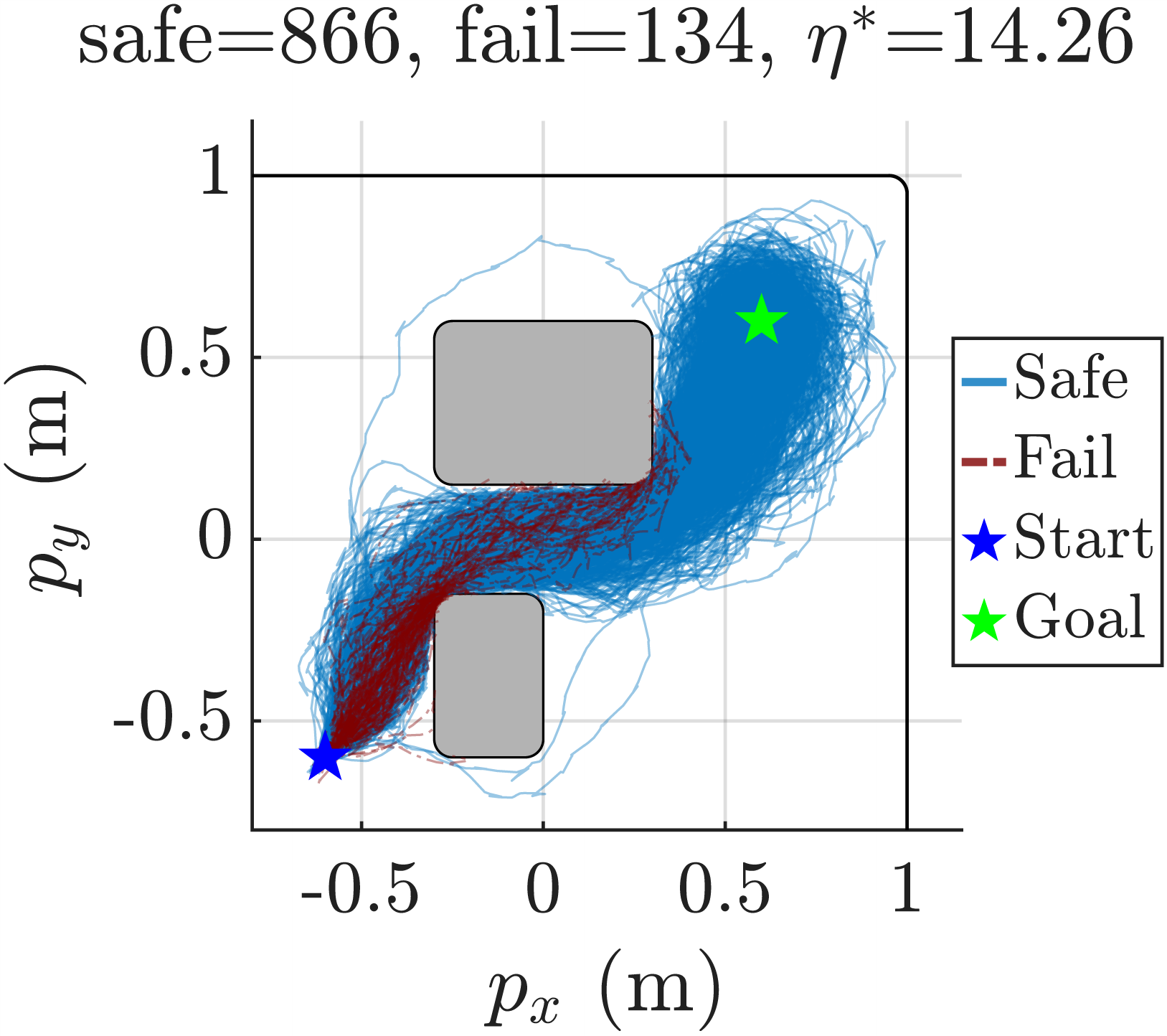}
\caption{Noncomm.: EM}
\label{fig:caseC_traj_noncomm_em}
\end{subfigure}
\hfill
\begin{subfigure}[b]{0.24\textwidth}
\includegraphics[width=\textwidth]{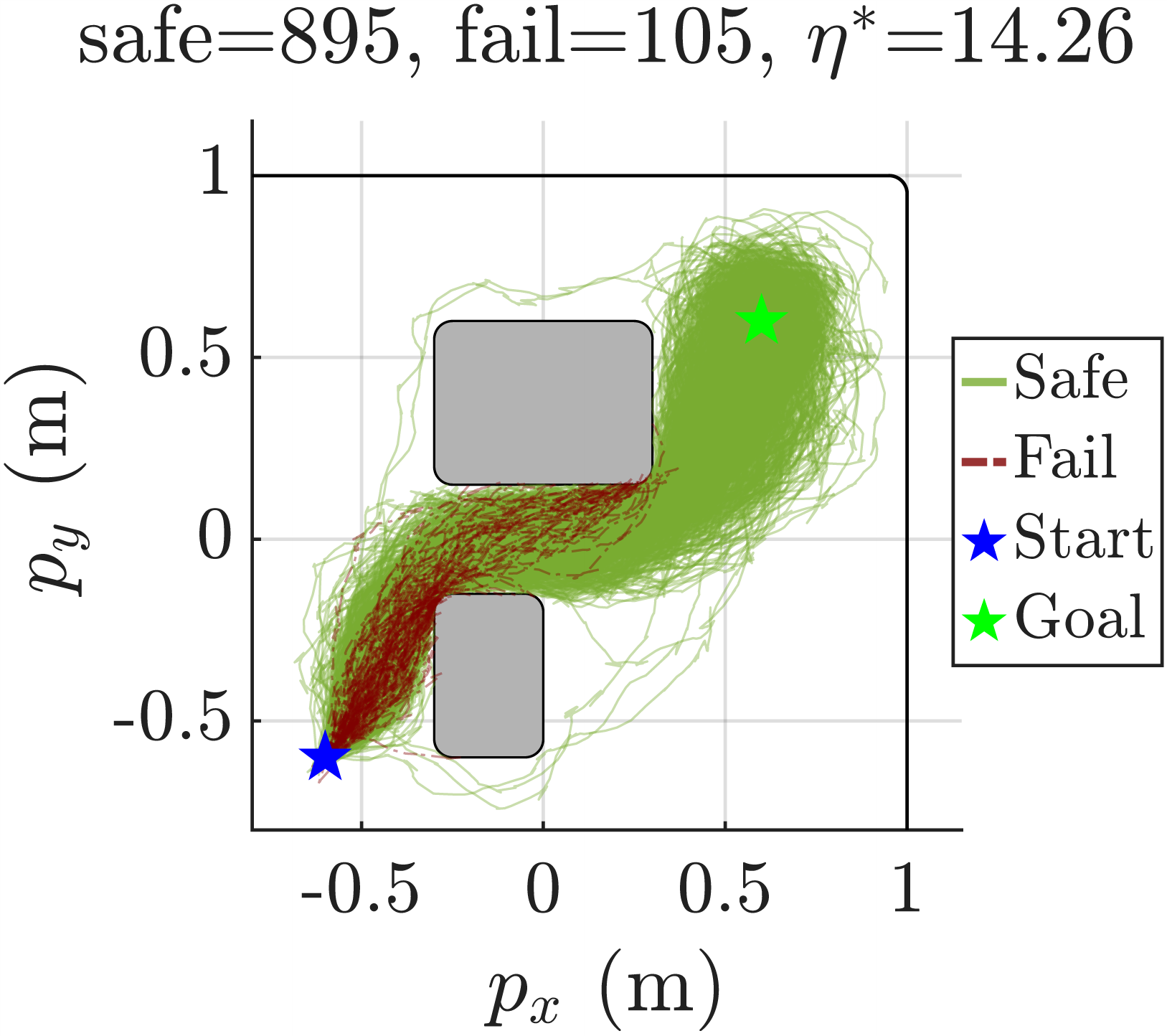}
\caption{Noncomm.: HR-W}
\label{fig:caseC_traj_noncomm}
\end{subfigure}
\hfill
\begin{subfigure}[b]{0.24\textwidth}
\includegraphics[width=\textwidth]{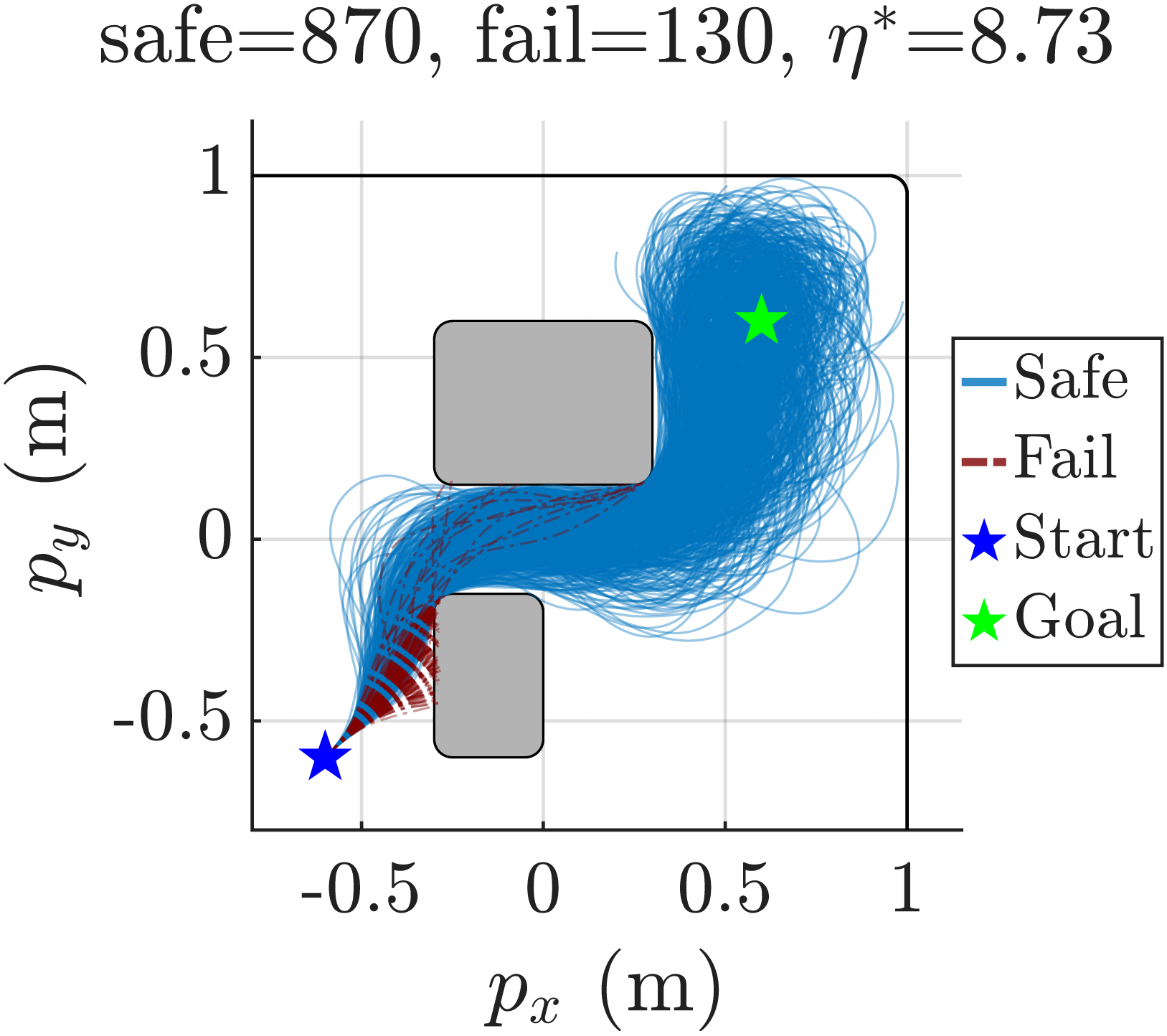}
\caption{Comm.: EM}
\label{fig:caseC_traj_comm_em}
\end{subfigure}
\hfill
\begin{subfigure}[b]{0.24\textwidth}
\includegraphics[width=\textwidth]{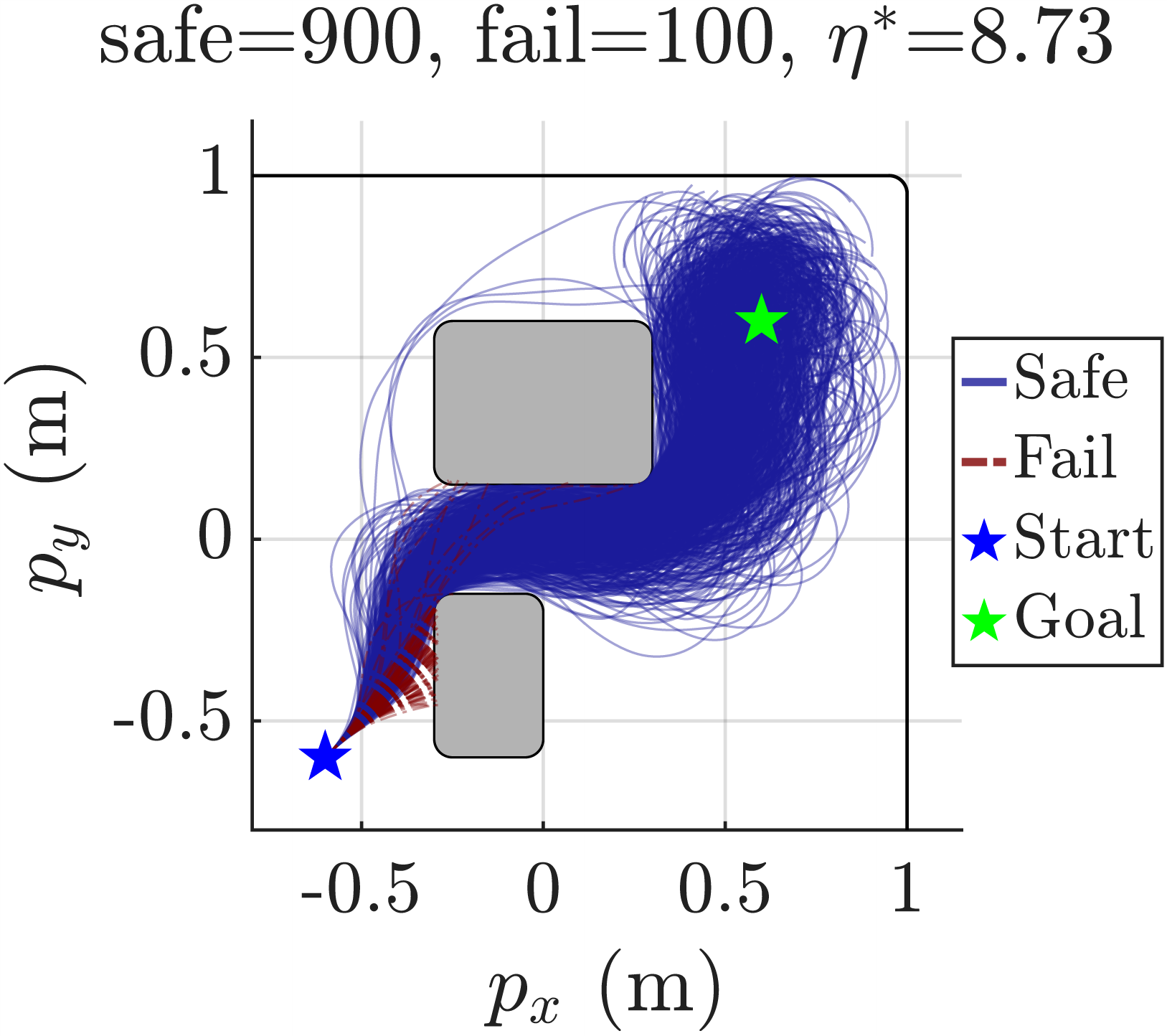}
\caption{Comm.: HR-1.5}
\label{fig:caseC_traj_comm}
\end{subfigure}
\caption{Closed-loop trajectories with gray obstacles. EM (left, blue) has higher failure probability than higher-order methods (right).}
\label{fig:caseC_trajectories}
\end{figure*}

Three experiments evaluate the preceding rates on a unicycle system. Cases~A and~B are offline strong-error benchmarks verifying the convergence rates of Theorems~\ref{thm:generic_indicator_rate}, \ref{thm:noncomm_order1_exit}--\ref{thm:comm_order15_exit}. Case~C evaluates closed-loop constraint satisfaction with CC-PIC.

\subsection{Benchmark models}
\label{subsec:benchmark_models}

We consider two planar unicycle models and two safe sets.

\paragraph{Model A: Noncommutative unicycle.}
In Model~A, the state is $\boldsymbol{X}=[p_x,\,p_y,\,\theta]^\top\in\mathbb{R}^3$ with driftless autonomous dynamics ($\mathbf{F}=\mathbf{G}\mathbf{u}$ in~\eqref{eq:controlled_sde}), given by
\begin{align}
\label{eq:unicycle_dynamics}
d\boldsymbol{X} &=
{\begin{bmatrix} \cos\theta & 0 \\ \sin\theta & 0 \\ 0 & 1 \end{bmatrix}}
{\begin{bmatrix} v \\ \omega \end{bmatrix}}
+ {\begin{bmatrix} \sigma_v\cos\theta & 0 \\ \sigma_v\sin\theta & 0 \\ 0 & \sigma_\omega \end{bmatrix}}
\end{align}
where $\mathbf{u}=[v,\,\omega]^\top$ is the control (linear velocity and yaw rate), with nominal values $v=0.6$, $\omega=0.25$, $\sigma_v=0.35$, $\sigma_\omega=0.80$, and $\boldsymbol{W}\in\mathbb{R}^2$. The matching condition~\eqref{eq:matching} is satisfied with $\mathbf{B}=\mathbf{G}\,\mathrm{diag}(\sigma_v,\sigma_\omega)$, which fixes $\mathbf R=\lambda\,\mathrm{diag}(\sigma_v^{-2},\sigma_\omega^{-2})$. The system is noncommutative, since $L_2 \mathbf{b}_1 = \sigma_\omega\sigma_v[-\sin\theta,\,\cos\theta,\,0]^\top \neq 0$ while $L_1 \mathbf{b}_2 = 0$, so accurate strong approximation beyond order~$0.5$ requires L\'evy-area simulation.

\paragraph{Model B: Lifted-actuator unicycle.}
Model~B promotes the controls $\mathbf{u}$ to states and places noise on the actuator channel. This lifting can apply to any controlled diffusion~\eqref{eq:controlled_sde} and recovers rate-based control as a special case~\cite{kim2022smooth,todorov2005stochastic}. The augmented state $\boldsymbol{Z}=(\boldsymbol{X},\mathbf{u})\in\mathbb{R}^{d+p}$ evolves as
\begin{align}
\label{eq:lifted_general}
d\boldsymbol{X} &= \big[\mathbf{f}(\boldsymbol{X}) + \mathbf{G}(\boldsymbol{X})\,\mathbf{u}\big]\,dt \nonumber\\
d\mathbf{u} &= \boldsymbol{\Lambda}(\mathbf{u})\,\boldsymbol\alpha\,dt + \sqrt{\lambda}\,\boldsymbol{\Lambda}(\mathbf{u})\,\mathbf{M}\,d\boldsymbol{W},
\end{align}
where $\boldsymbol\alpha$ is the new control input, $\boldsymbol{\Lambda}(\mathbf{u})=\mathrm{diag}(\gamma_i(u_i))$ with $\gamma_i(u_i)=\sqrt{\sigma_{0,i}^2+\sigma_{1,i}^2\,u_i^2}$, and $\mathbf{M}\mathbf{M}^\top=\mathbf{R}^{-1}$. The physical block is deterministic. Because each diffusion column depends only on its own coordinate $u_i$, the system is commutative, enabling the order-$1.5$ method without L\'evy areas. Applied to~\eqref{eq:unicycle_dynamics}, this gives $\boldsymbol{Z}\in\mathbb{R}^5$ with $\sigma_{0,v}=0.30$, $\sigma_{1,v}=0.10$, $\sigma_{0,\omega}=0.20$, $\sigma_{1,\omega}=0.10$.

\paragraph{Safe sets.}
Cases~A and~B use a disk of radius $R=1.0$ with a circular obstacle removed ($C^4$ boundary, $T=1.2\,\mathrm{s}$ and $T=5.0\,\mathrm{s}$, respectively, due to the slower exit rate of the lifted model). Case~C uses a rounded rectangle $[-1,1]^2$ with two rectangular obstacles forming a narrow passage ($r=0.05$ corner rounding, $T=7.0\,\mathrm{s}$). A superellipse approximation could recover $C^4$ regularity, but is not pursued here.

\paragraph{Strong-error methodology (Cases A and B).}
Errors are computed by paired-path Monte Carlo with $N=20{,}000$ samples sharing a common Brownian path. $21.8\%$ of reference paths exit in Case~A and $90.8\%$ in Case~B, ensuring nontrivial indicator statistics. The EM reference step size is $h_{\mathrm{ref}}=h_{\min}^2/10$ for Case~A and $h_{\mathrm{ref}}=h_{\min}^3/2$ for Case~B. 

\paragraph{Methods.}
For Model~A, Euler--Maruyama (EM), Milstein without L\'evy areas (Mil-noL), Milstein with Wiktorsson (Mil-W), adaptive HR without L\'evy areas (HR-noL), and adaptive HR with Wiktorsson (HR-W) were considered. For Model~B, EM, Milstein (Mil), and adaptive order-$1.5$ (HR-1.5) were considered.

\subsection{Case~A: Noncommutative strong error}
\label{subsec:results_noncomm}

Figure~\ref{fig:strong_errors} reports the strong errors for both models. For the noncommutative model, HR-W achieves a fitted exit-time slope of $0.93$, approaching the order-$1$ rate of Theorem~\ref{thm:noncomm_order1_exit}, while EM achieves $0.48$ (order~$1/2$). The indicator error (HR-W: $0.92$, EM: $0.54$) is similarly consistent with Theorem~\ref{thm:generic_indicator_rate}.

\subsection{Case~B: Commutative strong error}
\label{subsec:results_comm}

For the commutative model, HR-$1.5$ achieves a fitted exit-time slope of $1.44$, approaching the theoretical $3/2$ rate (Theorem~\ref{thm:comm_order15_exit}), while EM and Mil both achieve $\approx 0.96$ (order~$1$). The indicator error shows HR-$1.5$ at $9.13$, effectively vanishing at these step sizes.

\subsection{Case~C: CC-PIC}
\label{subsec:results_pic}

The CC-PIC controller (Alg.~\ref{alg:ccpic}) is tested with $K=100{,}000$ rollouts, $K_{\mathrm{ctrl}}=50{,}000$ rollouts in closed loop, $\lambda=1.0$, $\Delta=0.10$, $h=0.05$, $T=7.0\,\mathrm{s}$, $N_{\mathrm{episodes}}=1{,}000$.

In practice, the dual ascent~\eqref{eq:dual_ascent} uses Adam in place of gradient ascent, preceded by a binary search to initialize $\eta$ near the optimal region. 
For the noncommutative model (Fig.~\ref{fig:caseC_traj_noncomm}), at the converged multiplier $\eta^*=14.26$, EM produces $13.4\%$ violations while HR-W achieves $10.5\%$, near the threshold $\Delta=0.10$. Violation rates for all methods appear in Fig.~\ref{fig:caseC_viol_noncomm}. For the commutative model (Fig.~\ref{fig:caseC_traj_comm}), at $\eta^*=8.73$, EM yields $13.0\%$ violations while HR-$1.5$ achieves exactly $10.0\%$, matching $\Delta$. Dual-ascent convergence and violation rates are shown in Figs.~\ref{fig:caseC_da_comm} and~\ref{fig:caseC_viol_comm}. 

\subsection{Comparison of discretization methods}
\label{subsec:route_comparison}

\begin{table}[H]
\centering
\caption{Comparison of exit-functional propagators.}
\label{tab:method_summary}
\footnotesize
\begin{tabular}{@{}lccc@{}}
\toprule
& \textbf{EM/Mil} & \textbf{HR-W} & \textbf{HR-1.5} \\
\midrule
Diffusion & Any & General & $[b_j, b_k]=0$ \\
Exit-time rate & $O(h^{1/2})$ & $O(h^{1-\xi})$ & $O(h^{3/2-\xi})$ \\
Indicator rate & $O(h^{\frac{\beta}{2(1+\beta)}})$ & $O(h^{\frac{(1-\xi)\beta}{1+\beta}})$ & $O(h^{\frac{(3/2-\xi)\beta}{1+\beta}})$ \\
Mesh scales & $\{h\}$ & $\{h,\,h^2\}$ & $\{h,\,h^2,\,h^3\}$ \\
L\'evy areas & None & $p\asymp h^{-1/2}$ & None \\
L\'evy-step cost & $O(1)$ & $O(m^2 h^{-1})$ & $O(1)$ \\
Total cost~\cite{hoel2024adaptive} & $O(h^{-1})$ & $O(h^{-1}\!\log h^{-1})$ & $O(h^{-1}\!\log h^{-1})$ \\
\bottomrule
\end{tabular}
\end{table}

Table~\ref{tab:method_summary} summarizes the cases considered. When commutativity holds, the order-$1.5$ case is preferable (half-order improvement, no L\'evy areas, $\brate=1$ sufficient). For noncommutative systems, the order-$1$ Milstein with Wiktorsson improves over EM under $\brate>1$. In both cases, the adaptive methods produced tighter constraint satisfaction than fixed-step counterparts with marginal total cost rate overhead.

\section{CONCLUSION}
\label{sec:conclusion}
This paper studied higher-order exit-functional approximation in safe stochastic optimal control, where safety is enforced through exit events. Under a terminal boundary-layer estimate, higher order bounds for clipped exit-time and mesh-point state approximation were transferred to the failure indicator function and relaxed stopped cost. Two cases were studied. An adaptive order-$1$ Milstein with L\'evy areas for noncommutative dynamics, and adaptive order-$1.5$ for commutative dynamics, with chance-constrained path-integral control as the specialization. Numerical experiments confirmed predicted rates and showed tighter constraint satisfaction over Euler--Maruyama. Open directions include state-dependent exit penalties, sharper exit-functional rates, and investigation of multilevel Monte Carlo methods for estimation of exit functionals.

\bibliographystyle{IEEEtran}
\bibliography{ref}
\end{document}